\documentclass[11pt]{article}
\usepackage[margin=1in]{geometry}

\pdfoutput=1

\usepackage{balance}
\usepackage[utf8]{inputenc} % allow utf-8 input
\usepackage[T1]{fontenc}    % use 8-bit T1 fonts
\usepackage{lmodern}
\usepackage{url}            % simple URL typesetting
\usepackage{booktabs}       % professional-quality tables
\usepackage{amsfonts}       % blackboard math symbols
\usepackage{nicefrac}       % compact symbols for 1/2, etc.
\usepackage{microtype}      % microtypography
\usepackage{amsmath}
\usepackage{inputs/color-edits}
\usepackage{subcaption}

\usepackage{bm}
\usepackage{booktabs}
\usepackage{amssymb}
\usepackage{algorithm}
\usepackage[noend]{algpseudocode}
\usepackage{algorithmicx}
\usepackage{mathtools}
\usepackage{comment} % enables comment environment
\usepackage[most]{tcolorbox}
\usepackage{thm-restate}
\usepackage{xfrac}
\usepackage{multirow}
\usepackage[labelfont=bf, font=small,aboveskip=0pt,belowskip=0pt]{caption}
\usepackage{bm}
\usepackage{newfloat}
\usepackage{enumitem}
\usepackage{xpatch}
\usepackage{flushend}
\usepackage{amsthm}

\usepackage[labelfont=bf, font=small,aboveskip=0pt,belowskip=0pt]{caption}

\allowdisplaybreaks

\definecolor[named]{ACMBlue}{cmyk}{1,0.1,0,0.1}
\definecolor[named]{ACMYellow}{cmyk}{0,0.16,1,0}
\definecolor[named]{ACMOrange}{cmyk}{0,0.42,1,0.01}
\definecolor[named]{ACMRed}{cmyk}{0,0.90,0.86,0}
\definecolor[named]{ACMLightBlue}{cmyk}{0.49,0.01,0,0}
\definecolor[named]{ACMGreen}{cmyk}{0.20,0,1,0.19}
\definecolor[named]{ACMPurple}{cmyk}{0.55,1,0,0.15}
\definecolor[named]{ACMDarkBlue}{cmyk}{1,0.58,0,0.21}
\usepackage[bookmarksnumbered,unicode]{hyperref}
\hypersetup{colorlinks,
  linkcolor=ACMRed,
  citecolor=ACMPurple,
  urlcolor=ACMDarkBlue,
  filecolor=ACMDarkBlue}

\usepackage{zi4}
\usepackage[frozencache,cachedir=.]{minted}

\newtheorem{lemma}{Lemma}[section]

\newtheorem{corollary}[lemma]{Corollary}

\newtheorem{definition}[lemma]{Definition}

\newtheorem{invariant}{Invariant}

\addauthor{sb}{blue} % sb for Soheil
\addauthor{ld}{red} %ld for Laxman
\addauthor{kl}{purple} %kl for Kuba
\addauthor{he}{gray} %he for Hossein
\addauthor{ws}{orange} %ws for Warren
\addauthor{vm}{green} %vm for Vahab

\DeclareMathOperator*{\argmin}{arg\,min}
\DeclareMathOperator*{\argmax}{arg\,max}

\definecolor{mygreen}{RGB}{20,140,80}
\definecolor{linkcolor}{RGB}{0,0,230}
\definecolor{mylightgray}{RGB}{230,230,230}
\definecolor{verylightgray}{RGB}{245,245,245}

\algnewcommand{\IIf}[1]{\State\algorithmicif\ #1\ \algorithmicthen}
\algnewcommand{\EndIIf}{\unskip\ \algorithmicend\ \algorithmicif}

\newcounter{myalgctr}
\newtcolorbox{OuterBox}[1][]{%
    breakable,
    enhanced,
    frame hidden,
    interior hidden,
    left=-5pt,
    right=-5pt,
    top=-5pt,
    float=p,
    boxsep=0pt,
    arc=0pt
#1}%

\newtcolorbox{InnerBox}[1][]{%
    enforce breakable,
    enhanced,
    colback=gray,
    colframe=white,
#1}%

\newenvironment{tbox}{
\vspace{0.2cm}
\begin{tcolorbox}[
                  enhanced,
                  boxsep=2pt,
                  left=1pt,
                  right=1pt,
                  top=4pt,
                  boxrule=1pt,
                  arc=0pt,
                  colback=white,
                  colframe=black,
	              breakable,
	              floatplacement=t,
	              float
                  ]%%
}{
\end{tcolorbox}
}

\newcommand{\tboxhrule}[0]{\vspace{0.1cm} {\color{black} \hrule} \vspace{0.2cm}}

\newenvironment{titledtbox}[1]{\begin{tbox}#1 \tboxhrule}{\end{tbox}}

\newcommand{\myparagraph}[1]{\smallskip\noindent {\bf #1.}}

\newcommand{\id}[1]{\ifmmode\mathit{#1}\else\textit{#1}\fi}
\newcommand{\const}[1]{\ifmmode\mbox{\textc{#1}}\else\textsc{#1}\fi}
\newcommand{\degree}[1]{\ensuremath{d(#1)}}

\newcommand{\dynorient}{dynamic graph-orientation}
\newcommand{\trianglelink}{triangle-based linkage}
\newcommand{\bestedge}[1]{\ensuremath{\textsc{BestEdge}(#1)}}
\newcommand{\bestedgenoarg}{\ensuremath{\textsc{BestEdge}}}

\newcommand{\unionnoarg}[1]{\ensuremath{\textsc{Union}}}
\newcommand{\uniontext}[1]{union}
\newcommand{\merge}[1]{\ensuremath{\textsc{Merge}(#1)}}
\newcommand{\mergenoarg}{\ensuremath{\textsc{Merge}}}

\newcommand{\heap}[1]{\ensuremath{\textsc{Heap}(#1)}}
\newcommand{\nghheap}[1]{neighbor-heap}
\newcommand{\mergecost}[1]{\ensuremath{\mathsf{Cost}(#1)}}
\newcommand{\totaledges}[1]{\ensuremath{\mathcal{T}(#1)}}
\newcommand{\staleness}[1]{\ensuremath{\mathcal{S}(#1)}}

\newcommand{\meldheap}[1]{\ensuremath{\mathcal{H}(#1)}}
\newcommand{\meldtab}[1]{\ensuremath{\mathcal{Q}(#1)}}

\newcommand{\graphheap}{GraphHAC-Heap}
\newcommand{\heapbased}{heap-based}
\newcommand{\graphchain}{GraphHAC-Chain}
\newcommand{\chainbased}{chain-based}

\newcommand{\singlelink}{single-linkage}
\newcommand{\completelink}{complete-linkage}

\newcommand{\weightedavglink}{weighted average-linkage}
\newcommand{\unweightedavglink}{unweighted average-linkage}

\newcommand{\hacweight}[2]{\ensuremath{\mathcal{W}(#1, #2)}}

\newcommand{\storedwgh}[1]{\ensuremath{W_{\mathcal{S}}(#1)}}
\newcommand{\truewgh}[1]{\ensuremath{W_{\mathcal{T}}(#1)}}
\newcommand{\cut}[1]{\ensuremath{\mathsf{Cut}(#1)}}

\definecolor{revise-color}{rgb}{0.76, 0.23, 0.13}
\definecolor{mygreen}{rgb}{0.0, 0.5, 0.0}
\definecolor{amaranth}{rgb}{0.9, 0.17, 0.31}

\begin{document}\sloppy

%\title{Parallel Graph Algorithms in Constant Adaptive Rounds: \\ Theory meets Practice
\title{Hierarchical Agglomerative Graph Clustering\\ in Nearly-Linear Time\thanks{
This is the full version of a paper that will appear at ICML'21.
The authors can be contacted at laxmandhulipala@gmail.com, \{eisen, jlacki, mirrokni\}@google.com, and jeshi@mit.edu}
%This is the full version of a paper in PVLDB (to be presented at VLDB'21).
%The authors can be contacted at soheil@cs.umd.edu, laxman@mit.edu, and
%\{esfandiari, jlacki, mirrokni, wschudy\}@google.com}
}
\author{
  Laxman Dhulipala\\MIT CSAIL
  \and
  David Eisenstat\\Google Research
  \and
  Jakub Łącki\\Google Research
  \and
  Vahab Mirrokni\\Google Research
  \and
  Jessica Shi\\MIT CSAIL
}

%\author{
%  Soheil Behnezhad\\University of Maryland
%  \and
%  Laxman Dhulipala \thanks{Work done as a student researcher at Google Research.} \\MIT CSAIL
%  \and
%  Hossein Esfandiari\\Google Research
%  \and
%  Jakub Łącki\\Google Research
%  \and
%  Vahab Mirrokni\\Google Research
%  \and
%  Warren Schudy\\Google Research
%}

\maketitle
\begin{abstract}

We study the widely used hierarchical agglomerative clustering (HAC)
algorithm on edge-weighted graphs.
We define an algorithmic framework for hierarchical agglomerative graph clustering
that provides the first efficient $\tilde{O}(m)$ time
exact algorithms for classic linkage measures, such as complete- and WPGMA-linkage, as well as other measures.
Furthermore, for average-linkage, arguably the most popular variant of HAC, we provide an algorithm that runs in $\tilde{O}(n\sqrt{m})$ time.
For this variant, this is the first exact algorithm that runs in subquadratic time, as long as $m=n^{2-\epsilon}$ for some constant $\epsilon > 0$.
We complement this result with a simple $\epsilon$-close approximation
algorithm for average-linkage in our framework that runs in $\tilde{O}(m)$ time.
As an application of our algorithms, we consider clustering
points in a metric space by first using $k$-NN to generate a graph from the point
set, and then running our algorithms on the resulting weighted graph.
We validate the performance of our algorithms on publicly available
datasets, and show that our approach can speed up clustering of point
datasets by a factor of 20.7--76.5x.

\end{abstract}

\section{Introduction}

Clustering is a fundamental and widely used unsupervised learning
technique with numerous applications in data mining, machine learning,
and social network analysis.
Hierarchical clustering is a popular approach to clustering which
outputs a hierarchy of clusters, where the input data objects are
singleton clusters at the bottom of the tree, with interior vertices
corresponding to merging the two children clusters.
In this paper, we consider the family of \emph{hierarchical agglomerative
clustering (HAC)} algorithms, which have attracted significant
theoretical and practical attention since they were first proposed nearly
50 years ago~\cite{king67stepwise, lance67general, sneath73numerical}.

A HAC algorithm takes as input a set of $n$ points and proceeds in $n-1$ steps.
In each step, it finds two most similar points and merges them together.
Here, the notion of similarity is defined by a configurable \emph{linkage measure}.

The specific choice of the linkage measure affects both the clustering quality and the computational complexity of the HAC algorithm.
In the case of a general linkage measure, HAC can be implemented in $O(n^3)$ time, assuming that we are given the similarity between each two points as input.
For the most commonly used linkage measure (single, average, Ward's and complete linkage) this complexity can be improved to $O(n^2)$ by using the nearest-neighbor chain algorithm~\cite{nn-chain}.

The $O(n^2)$ time algorithms for HAC are often referred to as optimal, given that they take the entire $n \times n$ similarity matrix as the input.
However, from a practical point of view, this quadratic lower bound is very pessimistic, as in many applications only a small fraction of the $n \times n$ similarities are non-negligible.
As an example, consider the problem of clustering search engine queries studied in~\cite{DBLP:conf/kdd/BeefermanB00}.
Each query is assigned a set of relevant URLs and the similarity between two queries is based on the overlap between their URL sets.
Clearly, for the vast majority of query pairs, the similarity is zero.

Knowing that in practice only $o(n^2)$ pairs of input points have non-negligible similarity scores results in two natural questions.
First, is it possible to design a subquadratic HAC algorithm in this case?
And if so, does this algorithm lead to improved running times in practical applications?
In this paper, we answer both these questions affirmatively.

\subsection{Our Contributions}
In this paper we study the HAC algorithm on edge-weighted graphs.
Formally, we consider a graph $G(V, E)$ with $n$ vertices and $m$ edges, where
each vertex represents one input point and edge weights describe the similarities between the endpoints.

We develop a general framework that encompasses both average-linkage, and complete- and WPGMA-linkage, where common primitives used in HAC are modularized into a neighbor-heap data structure, which offers tradeoffs in the theoretical guarantees depending on the representation used. In particular, different applications of this framework result in the first subquadratic (and in many cases near-linear time) algorithms for several variants of HAC.

\myparagraph{Complete-linkage and WPGMA-linkage}
As a direct application of our framework, we obtain $\tilde{O}(m)$-time exact algorithms for the complete-linkage and WPGMA-linkage measures.
For these measures we assume that the similarity between pairs of vertices not connected by an edge is ``undefined''. If we denote the undefined similarity by $\bot$, for each $x \in \mathbb{R}$ we have $\max(x, \bot) = x$ and $\mathcal{L}(x, \bot) = x$ for any linkage measure $\mathcal{L}$.

\myparagraph{Average-linkage}
As our main algorithmic result, we give an $\tilde{O}(n\sqrt{m})$ time-exact algorithm for HAC with the average-linkage (UPGMA) measure, as well as an $\tilde{O}(m)$ time approximate algorithm\footnote{Formally, for a constant $\epsilon$ we obtain an $\epsilon$-close algorithm according to the definition of~\cite{48657}.}.
Both algorithms assume that any two vertices not connected by an edge have a similarity equal to $0$.

Average-linkage is arguably the most commonly used variant of HAC, due to its very good empirical performance~\cite{doi:10.1021/ct700119m, moseley-wang, hac-reward}.
However, to the best of our knowledge, no subquadratic time algorithm for average-linkage HAC has been described prior to our work, even when $m = \Theta(n)$.

Our exact average-linkage HAC algorithm dynamically maintains a low-outdegree orientation of the graph, i.e., it assigns a direction to each edge, attempting to minimize the outdegrees of all vertices.
At each step, the maximum outdegree is $O(\alpha)$, where $\alpha$ is the \emph{arboricity} (see Section~\ref{sec:prelims} for the definition) of the graph, which allows us to bound the amortized number of updates per cluster merge by $O(\alpha + \mathrm{poly} \log n)$.

The running time of the exact algorithm is derived from a more general bound.
Consider a sequence of graphs $G_1 = G, G_2, \ldots, G_{n-1}$ computed by the HAC algorithm.
That is, $G_{i+1}$ is obtained by contracting the largest weight edge in $G_i$ and updating the edge weights using the chosen linkage measure.
We show that the exact algorithm runs in $\tilde{O}(n \cdot \bar{\alpha})$ time, where $\bar{\alpha}$ is the maximum arboricity of the graphs $G_1, \ldots, G_{n-1}$.
The $\tilde{O}(n \sqrt{m})$ bound comes from the fact that arboricity of any $m$-edge graph is upper-bounded by $O(\sqrt{m})$.
However, in many cases better bounds on graph arboricity are known.
In particular, on planar graphs (and more generally graphs with an excluded minor) our algorithm runs in $\tilde{O}(n)$ time.
Similarly, for graphs with treewidth at most $t$~\cite{robertson1986graph}, the overall running time is $\tilde{O}(nt)$.

\paragraph{Empirical evaluation.}
We provide efficient implementations of all the near-linear time algorithms that we give and study their empirical performance on large data sets.

Specifically we show an average speedup of 6.9x over a simple
exact average-linkage algorithm on large real-world graphs.
We have made our implementations publicly available on Github.
\footnote{\url{https://github.com/ParAlg/gbbs/tree/master/benchmarks/Clustering/SeqHAC}}

Finally, we show that our efficient implementations can be used to obtain a significantly faster HAC algorithm, even if the input is a collection of points.
Namely, we leverage an existing approximate nearest neighbor computation library~\cite{scann} to compute a k-nearest neighbor graph on the input data points.
We then cluster this graph using our efficient graph-based HAC implementation.
As we show, the end-to-end time of finding nearest neighbors followed by running our HAC implementation is 20.7-76.5x faster than an efficient $O(n^2)$ implementation on point sets.
At the same time, somewhat surprisingly, the quality of the clustering obtained by using our approximate method is on average the same as what is computed by the exact HAC, where for average-linkage, we achieve on average a 1.13\% increase on the Adjusted Rand-Index score and a 1.06\% increase on the Normalized Mutual Information score.

\subsection{Related Work}

The theoretical foundations of HAC algorithms have been well-studied \cite{Dasgupta2016, moseley-wang}, and have provided motivation for the use of certain variants of HAC in real-world settings \cite{RoyPokutta2016, CharikarChatziafratis2017, CoKaMa2017, ChChNi2019, hac-reward}.
Moreover, the version of HAC that takes a
graph as input has been studied before, especially in the context of graphs derived from point sets~\cite{GuRaSh1999, KaHaKu1999, knn-hac}, although without strong theoretical guarantees. 

A rich body of work has focused on breaking the quadratic time barrier for HAC.
In a recent major advancement, Abboud et al.~\cite{abboud19hac} showed that if
the input points are in Euclidean space and Ward's linkage method is
used~\cite{Ward1963}, only $\tilde{O}(n)$ similarity computations are needed.
By using an efficient nearest-neighbor data structure, they obtained a
subquadratic approximate HAC algorithm for Ward's linkage measure.

Another special case is the single-linkage measure, again in the setting of Euclidean space and using approximate distances.
By computing an approximate minimum spanning tree, one can obtain an approximate HAC solution that runs in $O(n \log n)$ total time~\cite{smid2018well}.
A related method is affinity clustering, which provides a
parallel HAC algorithm inspired by Bor\r{u}vka's minimum
spanning tree algorithm~\cite{bateni2017affinity}.

A number of papers proposed subquadratic HAC algorithms by sacrificing theoretical guarantees on the quality of the result, see~\cite{cochez2015twister} and the references therein.
A prominent line of work leverages locality sensitive hashing to obtain improved running bounds.
However, the improvements in the running time often come at a significant cost.
For example, some of the algorithms~\cite{cochez2015twister} do not come with a closed form expression giving the approximation ratio, or produce an incomplete dendrogram (by merging some datapoints together)~\cite{DBLP:conf/cikm/GilpinQD13}.

Additionally, many different HAC implementations are available, both open source and proprietary~\cite{wiki-hac}.
However, the fastest available implementations require $\Theta(n^2)$ time in the worst case (e.g.~\cite{fastcluster, scipy, sas-cluster}).
Such implementations are typically capable of clustering up to tens of thousands of points in few minutes, and in these cases they are heavily optimized, for example by taking advantage of GPUs~\cite{gpu-hac}. In comparison, by combining our efficient implementation with an efficient similarity search, we are able to cluster a dataset of almost $3\cdot 10^5$ points in under two minutes.

Finally, we note that although subquadratic approximate algorithms for HAC are known in a few settings, to the best of our knowledge, the implementations of these algorithms are not publicly available.

\iffalse
\emph{Graph Orientation}
Orientations have applications to fundamental sequential and parallel
graph algorithms, such as clique-counting and closely related
clique-based problems~\cite{chiba1985arboricity}. Low out-degree orientations, or $O(\alpha)$-outdegree orientations (also known as degeneracy orderings), are
total orderings on the vertices of a graph, where the oriented out-degree of each
vertex (the number of its neighbors higher than it in the ordering) is
bounded by $O(\alpha)$. These orientations have been widely used in classic graph algorithms including maximal independent set and clique-counting~\cite{Barenboim2010, shi2020parallel}.
\fi
\subsection{Preliminaries}\label{sec:prelims}

We denote a weighted graph by $G = (V, E, w)$, where $w : E \rightarrow \mathbb{R}$ assigns a weight to each edge.
The number of vertices in a graph is $n = |V|$, and the number of
edges is $m = |E|$.
When reporting asymptotic bounds, we assume that $m = \Omega(n)$.
Vertices are assumed to be indexed from $0$ to $n-1$.
Unless otherwise mentioned, all graphs considered in this paper are
undirected, and we assume that there are no self-edges or duplicate edges.
We use $N(v)$ to denote the neighbors of vertex $v$ and 
$\degree{v} = |N(v)|$ to denote its degree.
We use $\cut{X,Y}$ to denote the set of edges between two sets
of vertices $X$ and $Y$.

The \emph{arboricity} ($\alpha$) of a graph is the minimum number of
spanning forests needed to cover the graph.
Note that $\alpha$ is upper bounded by $O(\sqrt{m})$ and lower
bounded by $\Omega(m/n)$~\cite{chiba1985arboricity}.
A \emph{$c$-orientation} of a graph $G(V, E)$ directs each edge $e \in
E$ such that the maximum out-degree of
each vertex is at most $c$.
It is well known that every arboricity $\alpha$ graph admits an
$\alpha$-orientation.
The \emph{\dynorient{}} problem is to maintain an orientation of a graph 
as it is modified over a series of edge insertions and deletions.
The two quantities of interest are
the \emph{out-degree} of the maintained orientation and the
\emph{update time}, or the time taken to process each edge update.
We say that an algorithm maintaining a $c$-orientation in $\Delta$
time per update is a \emph{$(c, \Delta)$-\dynorient{}} algorithm.

\section{Graph-Based HAC}\label{sec:graph-hac}
Given a graph $G(V,E,w)$, the \emph{hierarchical agglomerative clustering (HAC)} problem for a
given linkage measure $\mathcal{L}$ is to compute a dendrogram by
repeatedly merging the two most similar clusters (the two clusters
connected by the largest-weight edge) until only a single cluster remains.
For simplicity, we assume that the graph contains a single connected
component, although our algorithms and implementations handle graphs
with multiple components.
We treat clusters and vertices interchangeably in this paper. For example,
we often refer to the degree $\degree{C}$ or neighbors $N(C)$ of a cluster $C$.
We use $\hacweight{C}{D}$ to denote the weight of an edge between two
clusters $C$ and $D$.
The size of a cluster $|C|$ is defined to be the number of initial 
(singleton) clusters that it contains.

Initially, there are $n$ singleton clusters containing the vertices
$v_0, \ldots, v_{n-1}$.
Clusters containing multiple vertices are generated over the course of
the algorithm by merging existing clusters.
The \emph{merge} of two clusters $X,Y$ results in a new cluster $Z$.
The weights of edges incident to the cluster formed by the merge are
given by a \emph{linkage measure} (discussed below).
A \emph{dendrogram} is a vertex-weighted tree where the leaves are the
initial clusters of the graph, the internal vertices correspond to
clusters generated by merges, and the weight of a vertex created by
merging two clusters $X,Y$ is given by the weight (similarity) of the 
edge between the two vertices at the time they are merged.

\myparagraph{Linkage Measures}
A linkage measure specifies how to reweight edges incident to a 
cluster created by a merge.
Many different linkage measures that have previously been studied
are applicable in the graph-based setting.
In \emph{\singlelink{}}, the weight between two clusters $(X,Y)$ is
$\max_{(x,y) \in \cut{X,Y}} w(x,y)$, or the maximum-similarity edge between
two vertices in $X$ and $Y$.
In \emph{\completelink{}} the weight between two clusters $(X,Y)$ is 
$\min_{(x,y) \in \cut{X,Y}} w(x,y)$, or the minimum-similarity edge between 
two vertices in $X$ and $Y$.
In the popular \emph{\unweightedavglink{} (UPGMA-linkage)} measure (often called the average-linkage measure) the similarity between two clusters
$(X,Y)$ is $\sum_{(x,y) \in \cut{X,Y}} w(x,y)/(|X| \cdot |Y|)$
or the total weight of inter-cluster edges between $X$ and $Y$, normalized by the number of possible inter-cluster edges.
The \emph{\weightedavglink{} (WPGMA-linkage)} measure is similar, but is
defined in terms of the current weights. 
If a cluster $Z$ is created by merging clusters
$X, Y$, the similarity of the edge between $Z$ and a neighboring cluster
$U$ is $\frac{\hacweight{X}{U} + \hacweight{Y}{U}}{2}$ if both the $(X,U)$
and $(Y,U)$ edges exist, and otherwise just the weight of the existing edge.

We define the \emph{best-neighbor} of a cluster $X$ to be the cluster $\argmax_{Y \in N(X)} \mathcal{W}(X,Y)$.
We call the edge connecting $X$ and $Y$ the \emph{best-edge} of $X$.
We say that a linkage measure is \emph{reducibile}~\cite{nn-chain}, if
for any three clusters $X,Y,Z$ where $X$ and $Y$ are mutual best-neighbors, it holds that $\mathcal{W}(X \cup Y, Z) \geq \max(\mathcal{W}(X,Z),
\mathcal{W}(Y,Z))$.

Our framework yields exact HAC algorithms for any reducible linkage
measure that also satisfies the following property.

\begin{definition} \label{def:trianglelink}
A linkage measure is called \emph{triangle-based} if it satisfies the following property.
Consider any step of the algorithm which merges clusters $B$ and $C$ into a cluster $B \cup C$.
Let $A$ be a cluster distinct from $B$ and $C$.
Then, if edge $(A, C)$ does not exist, $\hacweight{A}{B} = \hacweight{A}{B\cup C}$.
\end{definition}

In other words, the weight of an edge in an \trianglelink{}
measure changing implies that the affected cluster (a cluster not participating
in the merge) was part of a triangle with both of the clusters being
merged.

\begin{restatable}{lemma}{trianglelinkages}
Single-, complete-, and WPGMA-linkage are all triangle-based linkages.
\end{restatable}

Unfortunately, while \unweightedavglink{} is a reducible linkage measure, 
\emph{it is not a \trianglelink{}}.
We design a special algorithm
for \unweightedavglink{} in Section~\ref{sec:averagelink}.

\section{Algorithmic Framework}\label{sec:framework}

In this section we design an algorithmic framework for solving
graph-based HAC for \trianglelink{} measures. 
We give two different algorithms, based on the nearest-neighbor chain and heap-based algorithms respectively from the classic literature on HAC.

\myparagraph{Overview}
There are two key substeps found in both the classic nearest-neighbor chain and heap-based algorithms: 
(i) \emph{merging} two clusters to obtain a new cluster and 
(ii) finding the \emph{best-edge} (edge to the most similar neighbor) out of a given cluster.
The classic algorithms use simple ideas to implement both steps, for example, by
using a linear-time merge for step (i), or by checking the similarity between
all pairs of points for step (ii) in the case of the chain-based algorithm.

Unfortunately, directly applying these simple ideas to graphs yields algorithms with quadratic time-complexity.
For example, implementing (i) using a linear-time merge algorithm will,
on a $n$-vertex star graph $(m = O(n))$, take $\Theta(n^2)$ time for any sequence of merges.
A similar example shows that exhaustively searching a neighbor-list for step (ii) in
a chain-based algorithm may take up to $\Theta(n^2)$ time.

We address these problems by noting that we 
can reuse the data structures for clusters that are being merged,
and by using heap data structures that support efficient updates.
For example, when merging two clusters,
we can reuse data structures associated with both merged clusters since 
they are logically deleted after the merge.
Furthermore, for efficiency we represent neighbors of a cluster 
using data structures that can merge two instances of size 
$s, l $ where $s \leq l$ in $O(s\log{(l/s + 1)})$ time.
Our analysis shows that we can perform any sequence of merges, while 
performing best-edge queries on the intermediate graphs in $\tilde{O}(m)$ time.

\subsection{Algorithms}

\begin{algorithm}[!t]\caption{$\textsc{Merge}(A, B, \mathcal{L})$}\label{alg:merge}
    \begin{algorithmic}[1]
        \Require{Active clusters $A$ and $B$, \trianglelink{} $\mathcal{L}$.}
        \Ensure{Cluster id of the remaining active cluster.}
        \State $(A,B) = (\argmin(\degree{A}, \degree{B}), \argmax(\degree{A}, \degree{B})$\label{line:mergesetsmaller}.
        \State Remove $B$ from \heap{A} (and vice versa).
        \State $\heap{B} = \heap{A} \cup \heap{B}$ (using $\mathcal{L}$ to merge
        the weights of $C \in \heap{A} \cap \heap{B}$). \label{line:mergemerge}
        \State Mark cluster $A$ as inactive. \label{line:mergeinactive}
        \State For each $C \in \heap{A}$, update the cluster-id from $A$ to $B$ in \heap{C}, using $\mathcal{L}$ to merge if $B \in \heap{C}. $\label{line:mergemapsmaller}
        \State Return $B$.
    \end{algorithmic}
\end{algorithm}

\myparagraph{Data Structures and Common Primitives}
Initially all $v \in V$ are \emph{active} singleton clusters.
Each cluster $A$ maintains a \emph{\nghheap{}} data structure \heap{A},
which is abstractly a max-heap that stores the neighbors 
of cluster $A$. The heap elements are key-value pairs containing the 
neighbor's cluster id and the weight (similarity) to the neighbor. 
The \nghheap{} supports the \bestedgenoarg{} operation, which returns
the best (highest-priority) edge in $H$.
In addition, the data structure supports the \unionnoarg{} operation, 
which given two \nghheap{}s $H_1, H_2$, and a \trianglelink{} $\mathcal{L}$, creates 
$H_1 \cup H_2$, where the weights of the pairs in $\textsc{Keys}(H_1) \cap \textsc{Keys}(H_2)$ 
are merged using $\mathcal{L}$.

We consider two different representations of a \nghheap{}. 
The first is a deterministic representation using
augmented balanced trees where 
the augmented values are the heap priorities~\cite{blelloch16justjoin}.
We also consider a representation of \nghheap{}s using
mergeable heaps (e.g., Fibonacci heaps) combined with
hash tables, which enables a faster implementation of 
the chain-based algorithm (discussed further in the appendix).
If $s=\min(|H_1|,|H_2|), l = \max(|H_1|, |H_2|)$, the cost of \unionnoarg{} is $O(s\log(l/s + 1))$ using
augmented trees, and $O(s)$ amortized using mergeable heaps.
The cost for \bestedgenoarg{} in both implementations is $O(\log n)$.

Merging in both of our algorithms is performed using Algorithm~\ref{alg:merge}. 
Given two clusters $A$ and $B$, suppose without loss of generality that $A$ 
has fewer neighbors. The algorithm merges $A$ into $B$, where 
the weights to neighbors in $N(A) \cap N(B)$ are computed using the linkage
measure $\mathcal{L}$ (Lines~\ref{line:mergesetsmaller}--\ref{line:mergemerge}).
$A$ is then marked as inactive (Line~\ref{line:mergeinactive}).
Lastly, the algorithm maps over each neighboring cluster $C \in N(A)$,
deletes the entry for $A$ in \heap{C}, and inserts this entry with the new 
cluster ID $B$, merging using $\mathcal{L}$ if $B$ already exists in \heap{C}.
Critically, this algorithm ensures that after it runs, all clusters which previously
had edges to $A$ now point to $B$, and that the weights of all edges to $B$
are correctly updated using $\mathcal{L}$.

Next, we provide pseudocode for the two HAC algorithms in our framework.
Note that both algorithms output a dendrogram $D$, but for simplicity 
we do not show the pseudocode of this step 
(the dendrogram can easily be maintained as part of Algorithm~\ref{alg:merge}).

\begin{algorithm}[!t]\caption{\graphchain{}($G=(V, E, w), \mathcal{L}$)}\label{alg:nnchain}
    \begin{algorithmic}[1]
        \Require{Edge weighted graph, $G$, \trianglelink{} $\mathcal{L}$.}
        \Ensure{Dendrogram $D$ for $\mathcal{L}$-HAC.}
        \For{each cluster $v \in V$}
          \If {$v$ is active}
            \State Initialize stack $S$, initially containing only $v$.
            \While {$S$ is not empty}
              \State Let $t$ be $\textsc{Top}(S)$.
              \State Let $b = \textsc{BestEdge}(t)$.\label{line:nnchainbestedge} 
              \If {$b$ is already on $S$}
                \State $\textsc{Pop}(S)$.
%                \kuba{To address tiebreaking we can say: Pop from S, Merge(w, Top(S)), Pop from S}
                \State $\textsc{Merge}(t, \textsc{Top}(S))$.\ \ 
                [Algorithm~\ref{alg:merge}] \label{line:nnchainmerge}
                \State $\textsc{Pop}(S)$.
              \Else
                \State Push $b$ onto $S$.
              \EndIf
            \EndWhile
          \EndIf
        \EndFor
    \end{algorithmic}
\end{algorithm}
Algorithm~\ref{alg:nnchain} shows the pseudocode for our \chainbased{} algorithm.
The structure of our algorithm is similar to the classic nearest-neighbor
chain algorithm~\cite{murtagh1983survey}, using a stack to maintain a path of best-neighbors and merging two
vertices that are connected by a reciprocal best-edge.

Algorithm~\ref{alg:heapbased} shows the pseudocode for our \heapbased{} algorithm.
Our algorithm uses a \emph{lazy} approach for handling edges in 
the global heap $H$ which point to inactive clusters (Lines~\ref{line:heapcheckinactive}--\ref{line:heapinsertfirst}).
Although we could potentially eagerly update the heap $H$ when deactivating vertices
in Algorithm~\ref{alg:merge} without an asymptotic increase in the running time, 
the lazy version is simpler to describe.

\subsection{Analysis}
We start with the following lemma, which bounds the cost of 
the \mergenoarg{} operations performed in Algorithm~\ref{alg:merge} for 
any sequence of $n-1$ merges.

\begin{restatable}{lemma}{lemmergecost} \label{lem:mergecost}
Let $m_1,\ldots, m_{n-1}$ be the
sequence of merge operations performed by an algorithm where $m_i = (u_i, v_i)$. 
Let $\mergecost{m_i} = \min(\degree{u_i}, \degree{v_i})$, where \degree{u_i} and \degree{v_i}
are the degrees of the clusters when they are merged. Then, the total cost $\sum_{i=1}^{n-1} \mergecost{m_i} = O(m \log n)$.
\end{restatable}

\begin{restatable}{theorem}{thmruntime}\label{thm:runtime}
There are deterministic implementations of the chain-based and heap-based algorithms
that run in $O(m\log^2 n)$ time for any \trianglelink{} $\mathcal{L}$.
\end{restatable}

\begin{algorithm}[!t]\caption{\graphheap{}($G=(V, E, w), \mathcal{L}$)}\label{alg:heapbased}
    \begin{algorithmic}[1]
        \Require{Edge weighted graph, $G$, \trianglelink{} $\mathcal{L}$.}
        \Ensure{Dendrogram $D$ for $\mathcal{L}$-HAC.}
        \State Let $H$ be a max-heap storing the highest-weight edge incident to each active cluster in the graph.
        \While{$|H| > 1$}
          \State Let $e=(u, v, \hacweight{u}{v})$ be the max edge in $H$.
          \State Delete $e$ from $H$.
          \If {$v$ is inactive} \label{line:heapcheckinactive}
            \State Let $e' = (u, v', \hacweight{u}{v'}) = \bestedge{u}$. \label{line:heapbestfirst}
            \State Insert $e'$ into $H$. \label{line:heapinsertfirst}
          \Else
            \State $x = \textsc{Merge}(u, v)$.\ \ [Algorithm~\ref{alg:merge}] \label{line:heapmerge}
            \State Let $e' = (x, y, \hacweight{x}{y}) = \bestedge{x}$. \label{line:heapbestsecond}
            \State Insert $e'$ into $H$.
          \EndIf
        \EndWhile
    \end{algorithmic}
\end{algorithm}

\begin{proof}[Proof sketch]
We sketch the proof for the chain-based algorithm.
The complete proof is given in the appendix.

Clearly, there are at most $n-1$ \mergenoarg{} operations and the total number of merged elements is bounded by $O(m \log n)$ due to Lemma~\ref{lem:mergecost}.
Using an augmented tree, the amortized cost of merging a single element is bounded by $O(\log n)$, which results in the total time of $O(m \log^2 n)$.

From the properties of the nearest-neighbor chain algorithm  at most $2n$ elements are ever added to the stack, which implies that the \bestedgenoarg{} operation is called $O(n)$ times. Hence all \bestedgenoarg{} operations take  $O(n \log n)$ time.
\end{proof}

We show that in the case of the chain-based algorithm, we can obtain an
asymptotically faster algorithm by using hash tables and Fibonacci heaps to represent the \nghheap{}s. 

\begin{restatable}{theorem}{frameworkchainbased}
There is a randomized implementation of the chain-based algorithm
that runs in $O(m\log n)$ time in expectation for any \trianglelink{} $\mathcal{L}$.
\end{restatable}

\section{Average Linkage}\label{sec:averagelink}

The key challenge for HAC using the average-linkage (UPGMA-linkage)
measure is that merging two clusters into a new cluster affects the weights of 
\emph{all edges} incident to the new cluster, and thus our framework from
Section~\ref{sec:framework} is not directly applicable. 
We show that by carefully modifying our
framework, we can obtain an efficient exact algorithm for 
average-linkage that runs in sub-quadratic time.

\myparagraph{Overview}
Two simple ideas to support average-linkage in our framework are 
a \emph{fully eager} approach, which updates all of the edges 
incident to a merged cluster, and a \emph{fully lazy} approach,
which updates none of the edges incident to a merged cluster and forces
a \bestedge{v} computation to spend $d(v)$ time. Unfortunately, simple 
examples show that both approaches can be forced to spend $\Theta(n^2)$ time 
for an $m = O(n)$ edge graph (e.g., a star on $n$ vertices).

Our approach is to enable the HAC algorithm to perform \bestedgenoarg{}
queries while only updating a \emph{subset} of the edges incident to a
newly-merged cluster.
We achieve this by using a \dynorient{} data structure, which maintains
an dynamic $O(\alpha)$-outdegree orientation where $\alpha$ is the arboricity
of the current graph. 
Our observation is that
using this
data structure lets us maintain information about the 
current clustered graph using a \emph{bounded amount of laziness}.
Specifically, we maintain the invariant that each vertex stores the up-to-date weight of all its \emph{incoming} edges.

We show that we can perform the $i$-th merge with an extra
cost of $O(\alpha_i)$ (in addition to the merge-cost in the framework)
where $\alpha_i$ is the arboricity of the current graph at the $i$-th merge.
Similarly, we can perform a \bestedgenoarg{} operation in $O(\alpha_i)$ time.
Using the fact that $\forall i, \alpha_i \leq \sqrt{m}$, we can obtain a
sub-quadratic bound as long as the number of \bestedgenoarg{}
operations is $o(n^{2} / \sqrt{m})$.
Unfortunately, the heap-based algorithm could perform $O(m\log n)$ 
\bestedgenoarg{} computations, but we are guaranteed that the chain-based algorithm 
will only perform $O(n)$ \bestedgenoarg{} computations, and thus gives an algorithm
with $\tilde{O}(n\sqrt{m})$ time-complexity.

\begin{algorithm}[!t]\caption{$\textsc{FlipEdge}(A, B)$}\label{alg:flipedge}
    \begin{algorithmic}[1]
        \Require{Edge oriented from active clusters $A$ to $B$.}
        \Ensure{Edge oriented from $B$ to $A$.}
        \State $w = \hacweight{A}{B}$. [true weight of the edge]
        \State Update the edge $(A, B, w)$ in \heap{A}.
    \end{algorithmic}
\end{algorithm}

\begin{algorithm}[!t]\caption{$\textsc{UpdateOrientation}(A, B, \mathcal{EO})$}\label{alg:updateorientation}
    \begin{algorithmic}[1]
        \Require{Active clusters $A$ and $B$, dynamic orientation structure $\mathcal{EO}$.}
        \State $(A,B) = (\argmin(\degree{A}, \degree{B}), \argmax(\degree{A}, \degree{B})$.
        \State For each $C \in N(A)$, delete $(C, A)$ and insert $(C, B)$ into the orientation data structure. Edge flips are handled using \textsc{FlipEdge} [Algorithm~\ref{alg:flipedge}]
        \For {each edge $(B, C)$ oriented out of $B$}
            \State $w = \hacweight{B}{C}$. [true weight of the edge]
            \State Update the edge $(B, C, w)$ in \heap{C}.
        \EndFor
    \end{algorithmic}
\end{algorithm}

\myparagraph{Algorithm}
The differences between our exact average-link algorithm and 
Algorithm~\ref{alg:nnchain} are (1) that we maintain a \dynorient{}
structure $\mathcal{EO}$ and (2) that we run extra procedures 
before performing the \bestedgenoarg{} and \mergenoarg{} algorithms used by 
Algorithm~\ref{alg:nnchain}. We also make a minor change in how 
the weights of edges in each cluster's neighbor-heaps are stored.

\emph{Before a \textsc{Merge}($A$, $B$).} Before merging two active clusters
using \textsc{Merge} (Algorithm~\ref{alg:merge}) we call \textsc{UpdateOrientation}
(Algorithm~\ref{alg:updateorientation}). This algorithm updates $\mathcal{EO}$
by deleting all edges going to the smaller deactivated cluster ($A$), 
and relabeling and inserting these edges to refer to the remaining 
active cluster ($B$).
Note that the orientation data structure could cause a number of edges
to have their orientation \emph{flipped}. We inject Algorithm~\ref{alg:flipedge}, 
which is called each time an edge is flipped
and which updates the weight of 
the edge to its correct weight at the head of the new direction 
of the edge. The last step in Algorithm~\ref{alg:updateorientation} is
to update the weights for each of the edges oriented out of the active
cluster $B$ in the heap of this directed neighbor of $B$.

\emph{Before a \bestedge{A}.} Before extracting the best-edge from \heap{A},
the algorithm updates each of the edges currently oriented out of $A$ in
$\mathcal{EO}$. Specifically, for such a directed edge $(A, B)$, it
computes the true weight of this edge and updates this value in \heap{u}.
Performing this update is necessary, since $B$ could have updated its 
size since the last time the $(A,B)$ edge was updated in \heap{u}.

\emph{Weight Representation in Neighbor-Heaps.} If we
store the weights of edges in each \nghheap{}, then when a cluster's size increases 
through a merge, we must update all of the edges incident to the new cluster 
since the weights of all edges change.
Instead, we store only part of the edge weights. Specifically,
for an edge to $B$ incident to cluster $A$
we store $\frac{1}{|B|} \sum_{(a,b) \in \cut{A,B}} w(a,b)$
in \heap{A}, and implicitly multiply this quantity by $1/|A|$, explicitly 
multiplying by this quantity when extracting an actual weight.

We obtain our exact average-linkage algorithm using a 
specific \dynorient{} data structure. Specifically, we use the 
recent data structure of~\cite{henzinger20dynamic} which maintains a $(O(\alpha_i), O(\log^2 n))$-\dynorient{}
data structure $\mathcal{EO}$. The out-degree of the orientation
maintained is \emph{adaptive}, i.e., the out-degree of vertices in $\mathcal{EO}$ 
after the $i$-th update is $O(\alpha(G_i))$ where $G_i$ is the graph at the time 
of the $i$-th update.

\begin{restatable}{theorem}{exactavglinkcomplexity}
The exact average-linkage algorithm is correct and runs in $\tilde{O}(n\sqrt{m})$ 
time for arbitrary graphs.
\end{restatable}
\begin{proof}
We provide a proof-sketch here, and provide a detailed proof in the appendix.
Note that the \dynorient{} data structure is only updated and used during the
\mergenoarg{} and \bestedgenoarg{} operations. Consider the $i$-th such operation, 
and let $G_i$ be the graph induced by the current clustering
at the time of this operation. 
It is easy to check that the cost for a \mergenoarg{} operation 
is $O(\alpha(G_i) + \mergecost{m_i}\log^2 n)$ 
and the cost for a \bestedgenoarg{} operation is $O(\log n + \alpha(G_i))$.
Therefore, the total number of updates to $\mathcal{EO}$ is $\sum_{i=1}^{n-1} \mergecost{m_i} = O(m\log n)$,
and the overall time-complexity of updating $\mathcal{EO}$ is $O(m\log^3 n)$.
Since $\forall i, \alpha(G_i) \leq \sqrt{m}$ and there are $O(n)$ merge
and best-edge operations, the time for updating the \nghheap{}s is
$O(n \sqrt{m} + n\log n)$. Thus the total time-complexity of the algorithm is
$\tilde{O}(n\sqrt{m} + m) = \tilde{O}(n\sqrt{m})$.
\end{proof}

Our approach yields near-linear time exact algorithms for graphs
whose minors all have bounded arboricity. Specifically, we show that for minor-closed
graphs, such as  planar graphs, bounded genus graphs, and bounded treewidth graphs we
obtain algorithms with $\tilde{O}(n)$ time-complexity.

\begin{corollary}\label{cor:minorclosed}
The exact average-linkage algorithm runs in $\tilde{O}(n \cdot t)$ time on any graph from a minor-closed graph family whose elements all have arboricity at most $t$.
\end{corollary}

\subsection{Approximation algorithm}
Next, we show that average-linkage can be approximated in nearly-linear time.
An \emph{$\epsilon$-close HAC algorithm} is an algorithm which only 
merges edges that have similarity at least 
$(1 - \epsilon) \cdot \mathcal{W}_{\max}$ where $\mathcal{W}_{\max}$
is the largest weight currently in the graph~\cite{48657}.
An $\epsilon$-close algorithm is constrained to merge
an edge that is ``close'' in similarity to the merge the exact algorithm 
would perform. At the same time, the definition
gives the algorithm flexibility in which edge it chooses to merge, which it
can exploit to save work.

The idea of our algorithm is to maintain an extra counter for each 
cluster which stores the size the cluster had at the last time 
that the algorithm updated \emph{all} of the incident edges of 
the cluster. Call this variable the \emph{staleness}, $\staleness{A}$, of a
given cluster $A$. Our algorithm maintains the following invariant:
\begin{invariant}\label{inv:staleness}
For any active cluster $A$, $|A| < (1+\epsilon)\staleness{A}$.
\end{invariant}
We maintain this invariant by checking after merging two clusters
if the size of the remaining active cluster $A$ is still smaller
than $(1+\epsilon)\staleness{A}$. If the invariant is violated, 
then the algorithm performs a \emph{rebuild} which updates the 
similarities of all edges incident to $A$ to their true 
weights. Note that these updates are performed on the \nghheap{}s
for both endpoints of the updated edges.

Since Invariant~\ref{inv:staleness} can be violated at most 
$O(\log_{1+\epsilon} n)$ times, any given edge will be processed 
during a rebuild at most $O(\log_{1+\epsilon} n)$ times over the course 
of the algorithm, for a total cost of $O(m \log_{1+\epsilon} n)$.

We note that we only apply this approximation idea with 
the \emph{heap-based} algorithm from our framework.
Although it may be possible to combine this notion of approximation 
with the chain-based algorithm, the analysis needs to handle the fact
that the chain-based algorithm merges local minima instead of global minima,
and so for simplicity we only consider the heap-based approach.

\myparagraph{Theoretical Guarantees}
To argue that our approximation algorithm yields a $\epsilon$-close HAC algorithm
for average-linkage, we show that the similarity of any edge stored 
within an active cluster's neighbors cannot be much larger than its 
true edge similarity. Let the stored similarity of an edge 
$(u,v)$ in the neighborhood be denoted \storedwgh{u,v} and 
the true similarity of this edge be \truewgh{u,v}.

\begin{restatable}{lemma}{weightsapproxclose}
Let $e=(U, V)$ be an edge in the neighborhood of an active cluster $U$. Then, $(1+\epsilon)^{-2}\storedwgh{U,V} \leq \truewgh{U,V}.$
\end{restatable}

This lemma implies that if we merge this edge, then in the worst case,
we will get a $(1+\epsilon)^{-2}$-close algorithm since the highest
similarity edge could have a true similarity of $\storedwgh{e}$, but
the edge merged by the algorithm could have a true similarity of 
$\truewgh{e} = (1+\epsilon)^{-2}\storedwgh{e}$. 
By setting the value of $\epsilon$ that we use internally appropriately, we
obtain the following result.

\begin{restatable}{theorem}{appxalgorithmbound}
There is an $\epsilon$-close HAC algorithm for the average-linkage
measure that runs in $O(m\log^2 n)$ time.
\end{restatable}
\section{Empirical Evaluation}

We implemented our framework algorithms in C++ using 
the Graph Based Benchmark Suite (GBBS)~\cite{dhulipala18scalable, dhulipala20grades},
and using the augmented maps from 
PAM~\cite{blelloch16justjoin} to represent \nghheap{}s.
We provide more details about our implementations in the appendix.
We run our experiments on a 72-core machine with $4\times 2.4\mbox{GHz}$
Intel 18-core E7-8867 v4 Xeon processors, and 1\mbox{TB} of memory.
Both GBBS and PAM are designed for parallel algorithms, but 
to enable a fair comparison with other sequential algorithms we disabled parallel execution.

\begin{figure}[t]
\begin{center}
%\vspace{-0.05in}
\includegraphics[scale=0.7]{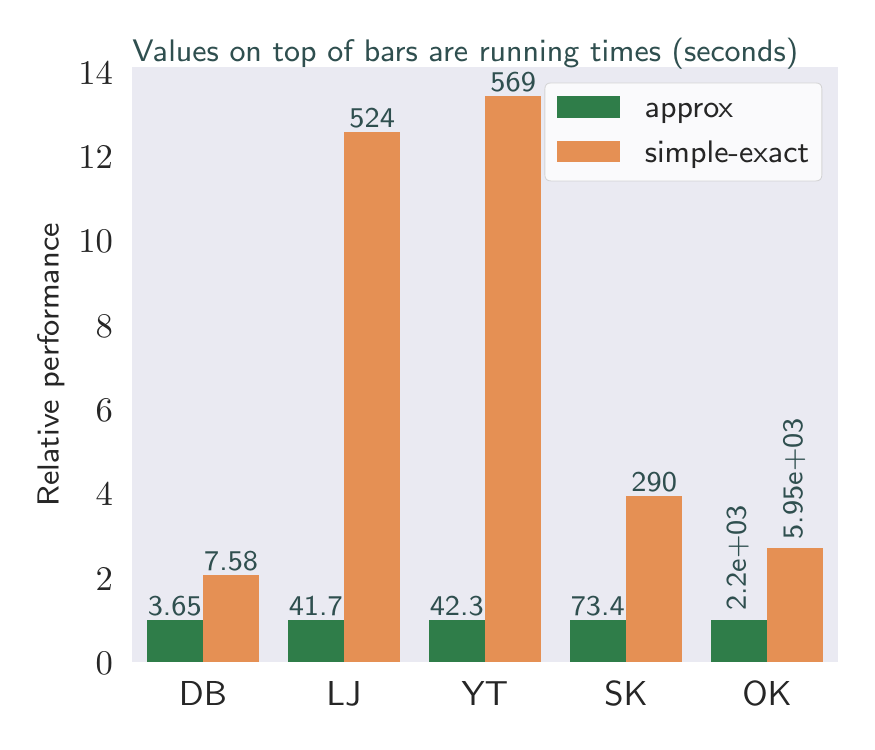}
%\vspace{-0.1in}
\caption{Relative performance of the approximate and 
simple-exact average-linkage algorithms on real-world graphs 
from the SNAP datasets, normalized to the fastest time per graph. 
The value on top of each bar is the running time in seconds.
\label{fig:avg_link_graphs}}
\end{center}
\vskip -0.1in
\end{figure}

\subsection{Experiments}

\myparagraph{Approximate vs. Simple-Exact Algorithm}
We start by evaluating the performance of our near-linear time
approximation algorithm for average-linkage vs.~a simple implementation of 
an exact average-linkage algorithm which updates the weights of all edges incident
to a newly merged cluster.
We ran this experiment on a collection of large real-world graphs from
the SNAP datasets. Since these graphs are originally unweighted, we set 
the similarity of an edge $(u,v)$ to $\frac{1}{\log(\degree{u} + \degree{v})}$.
We provide more details about our graph inputs in the appendix.
Figure~\ref{fig:avg_link_graphs} shows the result of the experiment.
Our approximate average-linkage algorithm (using $\epsilon = 0.1$) achieves an average speedup of 6.9x
over the exact average-linkage algorithm.
We note that the dendrograms in the cases where the simple-exact algorithm
performs reasonably well are shallower than in cases where the algorithm performs 
poorly (e.g., the DB dendrogram is 99x shallower than that of YT, although the
number of vertices in DB is only 2.6x smaller).
A simple upper-bound for the time-complexity of the simple-exact algorithm
is $O(nD)$ where $D$ is the depth of the dendrogram.
For other linkage measures, such as single-, complete-, and WPGMA-linkage,
we achieve up to 730x speedup over the simple-exact algorithm
that spends $O(d(u) + d(v))$ time to merge
two clusters, and note that the dendrograms observed for these measures have very high depth.

\myparagraph{Comparison with Metric Clustering}
Next, we study the quality and scalability of our graph-based HAC algorithms compared 
to metric HAC algorithms.
Given an input pointset, $P$, we first apply an approximate $k$-NN
algorithm to $P$ to build an approximate $k$-NN graph.
We use the state-of-the-art ScaNN $k$-NN library~\cite{scann} for 
this graph-building step.
We note that ScaNN internally uses multithreading, which we did not disable.
We then symmetrize the $k$-NN graph and run our graph-based HAC implementation on it.
We compare our results with those of the widely-used Scikit-learn (\emph{sklearn}) package.

\emph{Quality.} In the first set of experiments, we evaluate our algorithms and the four
HAC variants supported by sklearn on the \emph{iris}, \emph{wine}, \emph{digits},
and \emph{cancer}, and \emph{faces} classification datasets.
We note that the heap-based and chain-based algorithms yielded the same dendrograms.
To measure quality, we use the Adjusted Rand-Index (ARI) and Normalized Mutual Information (NMI)
scores.
The level of the tree with the highest score is used for evaluation.

We show the full quality scores in the appendix.
Overall, our graph-based algorithms produce results that are comparable with, 
and sometimes superior to the results of the metric-based algorithms in sklearn.
One exception is our complete-linkage algorithm, which is almost always worse
than the sklearn algorithm, which is because the $k$-NN graph
is missing large-distance edges which prevent cluster formation in 
the metric setting.
We note that running our complete-linkage algorithm on the complete graph
(with all distance edges) results in clustering results that match the
quality of the sklearn algorithm.
Our simple-exact and approximate average-linkage algorithm 
(using $\epsilon=0.1$) achieve essentially the same quality results 
as the exact sklearn algorithm (our algorithms achieve on 
average 1.8\% better ARI and 0.5\% better NMI). 
Furthermore, the approximate and simple-exact
algorithms yield identical quality results for all but
the digits dataset, where the simple-exact algorithm is 
less than 1\% better for both quality measures.

\begin{figure}[t]
\begin{center}
\includegraphics[scale=0.475]{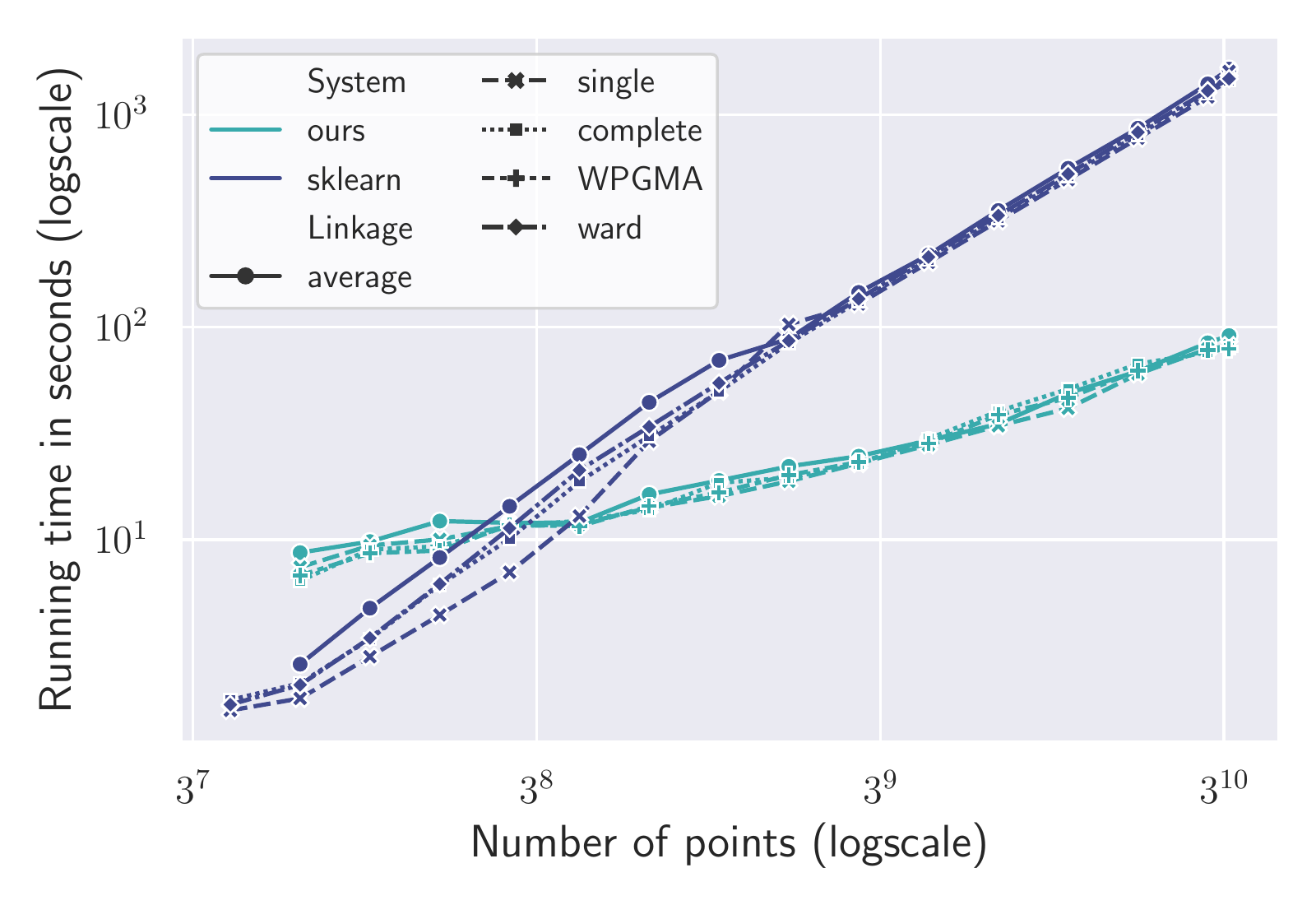}
\caption{End-to-end running times of the sklearn 
and our graph-based algorithms on varying-size slices of
the Fashion dataset.
\label{fig:fashion-comparison}}
\end{center}
\vspace{-0.1in}
\end{figure}

\emph{Scalability.} In the second set of experiments, we study whether our approach can
yield end-to-end speedups over the sklearn algorithms on large pointsets.
We use the Fashion-MNIST (764-dimensions), Last.fm (65 dimensions), and NYTimes (256 dimensions) 
datasets in these experiments.
We run both the sklearn and our algorithms on slices of these
datasets to understand how the running time scales as the number of points to cluster 
increases. 

Our results for the Fashion-MNIST dataset (shown in Figure~\ref{fig:fashion-comparison}) show that after about 10000 points, the end-to-end time of using
the graph-based approach is always faster than using the $O(n^2)$ time metric-based algorithm.
For the full Fashion-MNIST dataset, which contains 60,000 points, our approach yields an overall speedup
of 20.7x. 

\begin{figure}[t]
\begin{center}
%\vspace{-0.2em}
%\vspace{-0.04in}
\includegraphics[scale=0.475]{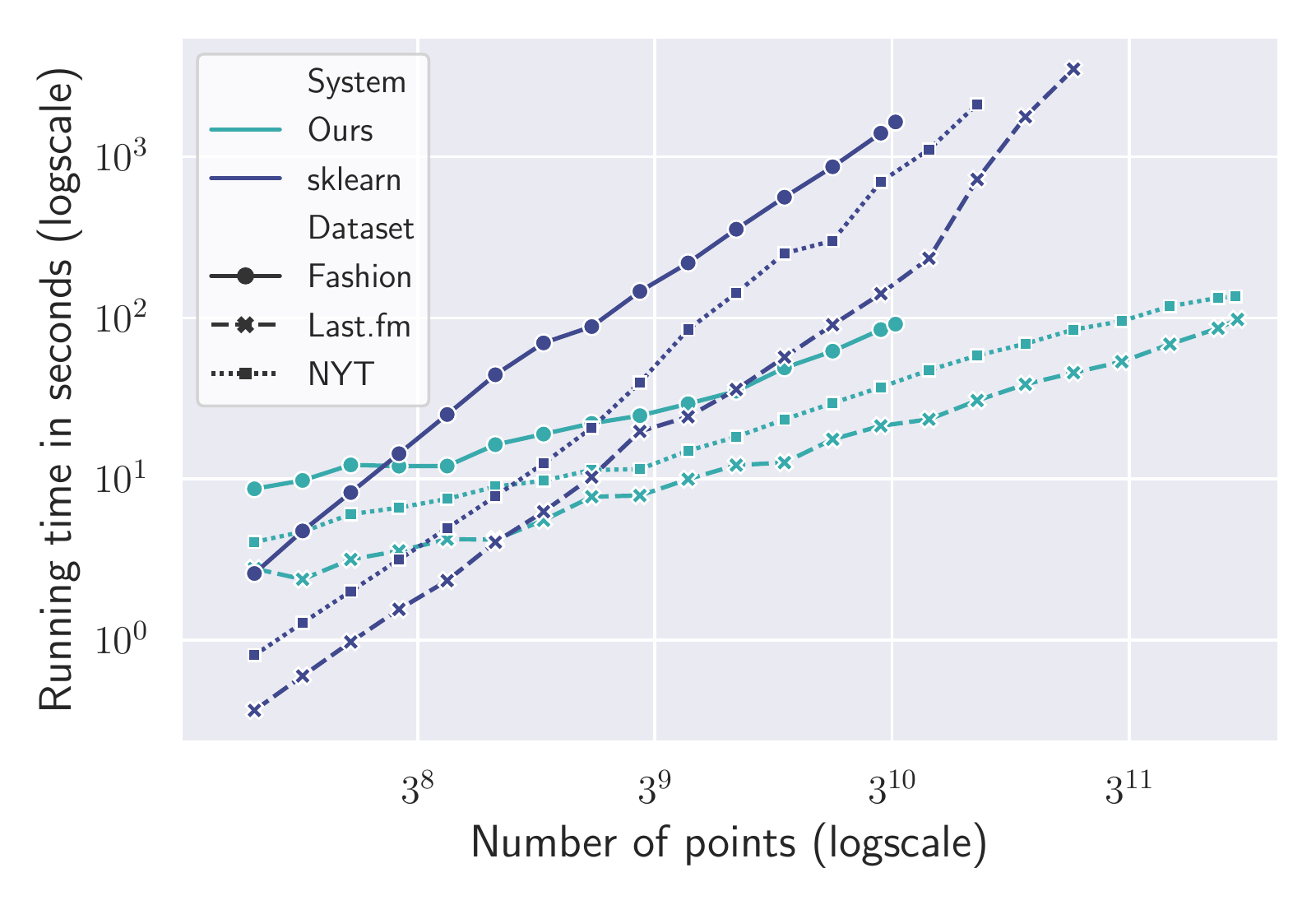}
%\vspace{-0.1in}
\caption{End-to-end running times of sklearn's
average-linkage and graph-based approximate average-linkage
on varying-size slices of
the Last.fm and NYTimes datasets.
\label{fig:pointset-comparison}}
%\vspace{-0.1in}
\end{center}
%\vspace{-0.3in}
\end{figure}

In Figure~\ref{fig:pointset-comparison} we show the results of the same experiment
using the Last.fm and NYTimes datasets, but using only our approximate average-linkage 
algorithm to reduce clutter (this is the slowest algorithm out of all of our 
linkage-measures). We terminated algorithms that ran for longer than
1 hour, and were therefore unable to finish running the metric-based algorithm on the 
full Last.fm (292,385 points) and NYTimes (290,000 points) datasets.
Our algorithms achieve up to 36.2x speedup on the NYTimes dataset and
76.5x speedup on the Last.fm dataset over the available datapoints for sklearn. 
Extrapolating from the trends of the sklearn implementations, a rough estimate suggests 
speedups of between 200x--500x for these datasets.

We note that we also ran the same experiments with the C++-based
HAC implementations provided in SciPy~\cite{scipy} a
nd Fastcluster~\cite{fastcluster}. We obtained running times that were 
within 10\% of the running times of sklearn for all of our datasets and
linkage functions, and thus only report the running times for 
sklearn in Figures~\ref{fig:fashion-comparison} and \ref{fig:pointset-comparison}.

\emph{Limits of our Approach.}
We observed that once we have generated a graph input, our algorithm's
performance scales almost linearly with the number of edges in the graph. 
Currently, the main bottleneck in our experiments for pointsets
is the graph-building step which generates the $k$-NN graph using ScaNN 
and writes the $k$-NN graph to disk.
Supplying the generated $k$-NN graph to our HAC algorithms without first
writing it to disk will further accelerate our algorithms.

Ignoring the cost of the writing to disk, both the memory usage
and running time of the graph-clustering step is lower than that of ScaNN.
Specifically, the memory usage of our algorithms (excluding ScaNN) is 
approximately $56\cdot m + O(n)$ bytes (where the constant on 
the $n$ term is small). Therefore, our implementations can solve 
graphs with up to several billion vertices and 2--3 billion edges 
on a machine with 256GB of memory (obtainable from Google Cloud for a few 
dollars per hour). In terms of running time, we observed that 
our graph-based algorithms scale linearly with the number of 
edges in practice, and we could thus solve such a graph in 
between 12--24 hours.

\section{Conclusion}
In this paper we designed efficient HAC algorithms which run in 
near-linear time with respect to the number of input similarity pairs.
We conducted a preliminary empirical evaluation, which shows 
that our algorithms achieve significant speedups while maintaining 
competitive clustering quality.

For future work, it would be very interesting to understand the parallel complexity of graph-based HAC, and to design efficient exact and approximate algorithms for these problems in a parallel or dynamic setting.
From an experimental perspective, a significant challenge is to design HAC
implementations that can be run on graphs with tens of billions of edges in a
reasonable amount of time.
Combining the ideas in this paper with an efficient dynamic graph
processing system such as Aspen~\cite{dhulipala19aspen} may be a first step 
towards such a result.
Finally, an interesting open question is to design a near-linear time 
exact HAC algorithm for the unweighted average-linkage measure.

\section*{Acknowledgements}
Thanks to D. Ellis Hershkowitz for helpful comments about this paper.
We would also like to thank the anonymous reviewers for their helpful
feedback and suggestions.

\bibliographystyle{abbrv}
\bibliography{references}

\appendix
\section{Algorithms and Definitions}\label{sec:rehash}

For ease of reference, we restate a few keys definitions and algorithms
that are used in the rest of the appendix.

\subsection{Linkage Measures}
Let the \emph{best-neighbor} of a cluster $X$ be $\argmax_{Y \in N(X)} \mathcal{W}(X,Y)$.
We call the edge connecting $X$ and its best-neighbor $Y$ the \emph{best-edge} of $X$.
We say that a linkage measure is \emph{reducible}~\cite{nn-chain}, if
for any three clusters $X,Y,Z$ where $X$ and $Y$ are mutual best-neighbors, it holds that $\mathcal{W}(X \cup Y, Z) \geq \max(\mathcal{W}(X,Z),
\mathcal{W}(Y,Z))$. Our framework yields exact HAC algorithms for any
\emph{reducible} linkage measure that also satisfies the following
property.

\begin{definition} \label{def:trianglelink}
A linkage measure is called \emph{triangle-based} if it satisfies the following property.
Consider any step of the algorithm which merges clusters $B$ and $C$ into a cluster $B \cup C$.
Let $A$ be a cluster distinct from $B$ and $C$.
Then, if edge $(A, C)$ does not exist, $\hacweight{A}{B} = \hacweight{A}{B\cup C}$.
\end{definition}

\subsection{Algorithms}

Most of our algorithms use the Merge routine to merge two clusters,
which is given in Algorithm~\ref{alg:merge}.
We provide the two main HAC implementations in our framework,
which are based on the classic nearest-neighbor chain and
heap-based methods (Algorithm~\ref{alg:nnchain} and
Algorithm~\ref{alg:heapbased} respectively).

\section{Neighbor-Heap Details}\label{sec:datastruct}

\myparagraph{Augmented Heaps}
A simple implementation of \nghheap{}s is to store the edges incident
to a cluster using an augmented binary tree, or an augmented heap.
Given a binary tree storing key-value entries, and a
function $f$ taking an entry and yielding a real-valued priority, we
can obtain a max-heap (min-heap) by setting the initial augmented
value of each vertex to its priority calculated using $f$, and
inductively setting the augmented values for internal vertices using
$\max$ ($\min$) of their augmented value, and the augmented values of
their two children.  We refer to~\cite{pam} for details on
implementing augmented binary trees.

The \bestedgenoarg{} operation can be implemented by having the
augmented value at each tree vertex (corresponding to an edge of the graph) be the
edge weight, and the augmentation function to $\max$. Extracting the
best-edge can then be done using a find-like function which finds a
(key, value) pair in the tree that exhibits the overall augmented
value of the tree. Another way is to set the augmented value to a pair
of the edge weight and the neighbor id, and have the augmentation
function to take a lexicographic maximum. The best-edge in this
approach is simply the augmented value at the root of the tree. In our
implementations, we use the former approach, since it is more
space-efficient than the latter approach. Both implementations cost
$O(\log n)$ work per \bestedgenoarg{} operation.

The \unionnoarg{} operation can be implemented by using \unionnoarg{}
on the underlying augmented binary trees. The algorithm
from~\cite{blelloch16justjoin} (which is implemented in the PAM
library~\cite{pam}) merges two trees of size $n, m$ with $n \leq m$ in
$O(n\log (m/n + 1))$ time, which is asymptotically optimal for
comparison based algorithms~\cite{brown1979fast}.

\section{Deferred Proofs}\label{sec:proofs}

%\trianglelinkages*
\begin{restatable}{lemma}{trianglelinkages}
Single-linkage, complete-linkage, and WPGMA-linkage are all triangle-based linkages.
\end{restatable}
\begin{proof}
It is well known that all of these linkage rules are reducible (e.g., see ~\cite{irbook}).
To show that they also satisfy Definition~\ref{def:trianglelink}, observe that the
weight of an edge $(A,B)$ can change after merging clusters $B$ and $C$
only when $A$ is also connected to $C$. For example, for single-linkage, it is
easy to see that the edges that we take the $\max$ over when
calculating $\hacweight{A}{B \cup C} = \max_{(x,y) \in E(A,B \cup C)} w(x,y)$
can affect the result only when $\hacweight{A}{C} > \hacweight{A}{B}$. Similar calculations
show that complete-linkage and WPGMA-linkage are \trianglelink{}s.
\end{proof}

\subsection{Framework Analysis}

\begin{restatable}{lemma}{lemmergecost} \label{lem:mergecost}
Let $m_1,\ldots, m_{n-1}$ be the
sequence of merge operations performed by an algorithm where $m_i = (u_i, v_i)$.
Let $\mergecost{m_i} = \min(\degree{u_i}, \degree{v_i})$, where \degree{u_i} and \degree{v_i}
are the degrees of the clusters when they are merged. Then, the total cost $\sum_{i=1}^{n-1} \mergecost{m_i} = O(m \log n)$.
\end{restatable}
\begin{proof}
The proof is by a charging argument. Following the definition of $\mathsf{Cost}$, assign
a token to each edge in $N(\argmin(\degree{u_i}, \degree{v_i}))$ for each $i$. We
do not undercount the cost, since we assign $\mergecost{m_i}$ tokens at each step $i$
to the edges. We now bound the total cost by bounding the number of tokens that can
be assigned to an edge.

We conceptually add an extra \emph{total-edges} variable, \totaledges{u} to the data structures
storing the vertex neighborhoods, $N(u)$. This variable simply stores the total number
of edges that have been merged into this tree. At the start of the algorithm, $\totaledges{u} = d(u)$.
When two clusters $u$ and $v$ merge, if $v$ is the cluster remaining active after the merge,
$\totaledges{v}$ is incremented by $\totaledges{u}$.
Let $A$ be the set of active clusters. It is easy to check that $\sum_{c \in A} \totaledges{c} = 2m$
at all points in the algorithm.

Next, we bound the maximum number of tokens assigned to an edge by observing that each time
an edge has a token assigned to it in some step $i$, the total-edges of the set containing it doubles.
Since \totaledges{u} of a set $u$ can grow to at most $2m$, each edge can receive at
most $O(\log m) = O(\log n)$ tokens, and thus the total cost is $O(m\log n)$.
\end{proof}

\begin{restatable}{theorem}{thmruntime}\label{thm:runtime}
There are deterministic implementations of the chain-based and heap-based algorithms
that run in $O(m\log^2 n)$ time for any \trianglelink{} $\mathcal{L}$.
\end{restatable}
\begin{proof}
We first bound the time-complexity for the merge steps that both algorithms have in
common. Both algorithms perform $n-1$ merge operations, whose total
cost is $O(m\log n)$ using Lemma~\ref{lem:mergecost}.  To translate
this cost measure back to time-complexity, we note that the total time
taken for the $i$-th merge step using the deterministic \nghheap{}
implementation is at most $O(\mergecost{m_i}\log(n / \mergecost{m_i} +
1)) = O(\mergecost{m_i}\log n)$, and thus the total time to perform
all merges is $O(m\log^2 n)$.  We bound the time of the remaining
algorithm-specific steps separately.

\emph{Chain-based Algorithm.} We use a few well-known facts about this
algorithm, namely that (i) each of the $2n-1$ clusters that appears in
the dendrogram is pushed onto the stack exactly once, and (ii) the
number of times \bestedgenoarg{} is called is $O(n)$. Therefore, the
total time-complexity of these steps is $O(n\log n)$ and the overall
time-complexity is $O(m\log^2 n)$.

\emph{Heap-based Algorithm.} The remaining steps
involve extracting edges from the global heap $H$. We analyze two types of edges
that can be extracted: (i) edges whose remaining endpoint is inactive and
(ii) edges whose remaining endpoint is active. There are at most $n-1$
type (ii) edges, since each type (ii) edge results in a merge, and so the
time spent processing these edges is $O(n\log n)$.
Next, for the type (i) edges, observe that each such edge can be charged
to the deactivated endpoint, $v$, and that the total number of charges
for a cluster $v$ is at most its degree at the time it was deactivated.
Thus, the total cost for these edges is
$\sum_{i=1}^{n-1} \mergecost{m_i} = O(m\log n)$ (by Lemma~\ref{lem:mergecost}) and the time-complexity
for these steps is $O(m\log^2 n)$ since each inactive edge takes
$O(\log n)$ time to extract the current best edge from its active endpoint, and to
update $H$. Thus, the overall time-complexity is $O(m\log^2 n)$.
\end{proof}

\subsection{A Faster Randomized Chain-Based Algorithm}

In this section we present a randomized implementation of our
chain-based algorithm which runs in $O(m\log n)$ time.

\myparagraph{Overview}
There are two challenges posed by Algorithm~\ref{alg:merge} that we
must implement more efficiently in order to achieve an $O(m \log n)$
time HAC algorithm.
\begin{enumerate}
  \item The \emph{merge-cost} from Lemma~\ref{lem:mergecost} is
  $O(m\log n)$. Thus, in order to achieve $O(m\log n)$ time we must
  perform each merge operation in (amortized) $O(1)$ time per merged
  element. \label{challenge1}

  \item The overall algorithm also performs $O(m \log n)$
  \emph{neighbor-updates} in Line~\ref{line:mergemapsmaller} of
  Algorithm~\ref{alg:merge}, which remove the id of a merged vertex
  from an active neighbor's \nghheap{} and relabel it to the id of the
  new neighbor. Thus, we must either handle these updates lazily, or
  also handle them in (amortized) $O(1)$ time per operation.
  \label{challenge2}
\end{enumerate}

Our approach to handle (\ref{challenge1}) is to use a faster
randomized implementation of \nghheap{}s which we outline below. The
high-level idea is to use an efficient meldable heap, such as a
Fibonacci heap or Leftist heap in conjunction with hash tables. We deal with
(\ref{challenge2}) by \emph{eagerly} updating the hash tables of our
neighbors when performing a merge, but \emph{lazily} updating the IDs
stored in the meldable heap, except when we identify an edge that is
being merged. The overall cost of the hash table updates is $O(m\log
n)$ time in expectation. Although the updates to the heaps cost
$O(\log n)$ time each, since they require deleting two existing
elements and reinserting a new merged element, each of these updates
can be assigned uniquely to an edge in the original graph, and thus
the overall time complexity for these updates is also $O(m\log n)$.
We now provide a detailed description of our approach.

\subsubsection{Neighbor-Heaps using Meldable Heaps and Hashing}
We give an alternative implementation of \nghheap{}s, which is
asymptotically faster than the augmented-heap based implementation at
the cost of using randomization. The idea is to use a heap data
structure that supports efficient \emph{melding}, such as Fibonacci
heaps~\cite{fredman87}, in combination with a hash table.
The \nghheap{} representation for a cluster $v$ is a pair of a heap
and a hash-table where the neighbors of $v$ are stored in both data
structures.
Let $\meldheap{A}$ denote the heap and $\meldtab{A}$ denote the
table for a cluster $A$.
The elements in \meldheap{A} are pairs of a cluster-id, $C$, and the
associated weight of this edge \hacweight{A}{C}. The priority of an
element is just the weight.
The elements in \meldtab{A} are triples of a key (a cluster-id), $C$,
the associated weight of this edge \hacweight{A}{C}, and a pointer to
the location of the element for $C$ in \meldheap{A}.

The \textsc{Meld} operation on two mergeable heaps $\mathcal{H}_1,
\mathcal{H}_2$ can be done in $O(1)$ time. Note that this operation
\emph{does not} detect elements in $\textsc{Keys}(\mathcal{H}_1) \cap
\textsc{Keys}(\mathcal{H}_2)$, which is why we also store the elements
in both heaps in a hash-table (which implements intersection
efficiently).

We also define a \textsc{T-Merge} operation on two hash-tables in
\nghheap{}s, \meldtab{A}, \meldtab{B}, which works as follows. Without loss of generality let
$|\meldtab{A}| \leq |\meldtab{B}|$.
We map over the elements in \meldtab{A}, and insert them into the
larger size table. If an key $C$ appears in both tables, then we merge
this edge using the linkage function $\mathcal{L}$.  We also append
$C$ and the locations of $C$ in both heaps to an array $O$ that
collects all of the locations for $C \in \meldtab{A} \cap
\meldtab{B}$.
The \textsc{T-Merge} operation returns $O$ and \meldtab{B}.
\textsc{T-Merge} on two tables \meldtab{A}, \meldtab{B} runs in
$O(\min(|\meldtab{A}|, |\meldtab{A}|))$ expected time.

\subsubsection{Merging Clusters}
Next, we present how two clusters are merged using the randomized
\nghheap{} (Algorithm~\ref{alg:randomizedmerge}).
\begin{algorithm}\caption{$\textsc{FastMerge}(A, B, \mathcal{L})$}\label{alg:randomizedmerge}
    \begin{algorithmic}[1]
        \Require{Active clusters $A$ and $B$, \trianglelink{} $\mathcal{L}$.}
        \Ensure{Cluster id of the remaining active cluster.}

        \State $(A,B) = (\argmin(\degree{A}, \degree{B}), \argmax(\degree{A}, \degree{B}))$
        \State Remove $B$ from \meldheap{A} and \meldtab{A} (and vice versa). \label{randmerge:remove}

        \State $(O, \meldtab{B}) = \textsc{T-Merge}(\meldtab{A}, \meldtab{B})$
        \Comment{$O$ holds heap-locations for $C \in \meldtab{A} \cap \meldtab{B}$}  \label{randmerge:mergetables}

        \For {each $(C, L_A, L_B) \in O$}\label{randmerge:forstart}
          \State Delete $C$ from \meldheap{A} and \meldheap{B} using $L_A$ and $L_B$ to find these elements.\label{randmerge:deletelocs}
          \State Insert $C$ into \meldheap{B} with the weight merged using $\mathcal{L}$. Let the pointer to this element be $L'$.\label{randmerge:insertb}
          \State Update the location of $C$ in \meldtab{B} to $L'$.\label{randmerge:updateloc}
        \EndFor

        \State $\meldheap{B} = \textsc{Meld}(\meldheap{A}, \meldheap{B})$. \Comment{Before the meld, $\textsc{Keys}(\meldheap{A}) \cap \textsc{Keys}(\meldheap{B}) = \emptyset$.} \label{merge:meld}

        \For {$C \in \textsc{Keys}(\meldtab{A})$} \label{randmerge:fornghs}
          \State Update cluster-id from $A$ to $B$ in \meldtab{C} and \meldheap{C}. If $B \in \meldtab{C}$, use $\mathcal{L}$
            to merge the edge weights and update both \meldtab{C} and \meldheap{C}. \label{randmerge:updatengh}
        \EndFor

        \State Mark cluster $A$ as inactive.
        \State Return $B$.
    \end{algorithmic}
\end{algorithm}
Algorithm~\ref{alg:randomizedmerge} is similar to our original Merge
algorithm, Algorithm~\ref{alg:merge} with a few key differences.
First, Line~\ref{randmerge:remove} removes the IDs of the merged clusters from
both the heaps and hash-tables for each of the merging clusters.
Next, on Line~\ref{randmerge:mergetables} the algorithm merges the
hash-tables of both clusters using the \textsc{T-Merge} routine
described above. The result is a sequence $O$ of triples containing
the cluster-id, and two heap-locations of
$C \in \textsc{Keys}(\meldtab{A}) \cap \textsc{Keys}(\meldtab{B})$,
and the newly merged table, \meldtab{B}.
The algorithm then loops over these clusters $C$ with edges to both
$A$ and $B$, and the location of these edges in $\meldheap{A}$ and \meldheap{B} (Lines~\ref{randmerge:forstart}--\ref{randmerge:updateloc}).
For each such cluster, the algorithm first deletes $C$ from \meldheap{A}
and \meldheap{B} using the given locations (Line~\ref{randmerge:deletelocs}).
It then inserts $C$ into \meldheap{B} with the updated weight of this edge
(Line~\ref{randmerge:insertb}).
Note that at this point, $\textsc{Keys}(\meldheap{A}) \cap \textsc{Keys}(\meldheap{A}) = \emptyset$.
After the loop, the algorithm first melds the two heaps (Line~\ref{merge:meld}).
The last step is to update the \nghheap{}s of neighbors of $A$ (the deactivated cluster).
The algorithm iterates over all neighbors $C$ of $A$ on Line~\ref{randmerge:fornghs}, and
on Line~\ref{randmerge:updatengh}, updates the id from $A$ to $B$ in \meldtab{C} and \meldheap{C}. Note that the location
of $A$ in \meldheap{C} is stored \meldtab{C}. If $C$ also has an edge to $B$, it sets
the weight of the $(C,B)$ edge to the updated weight in \meldtab{C}. It also deletes
$A$ and $B$ from \meldheap{C} and reinserts $B$ into \meldheap{C} with the correct weight.
Finally, it marks $A$ as inactive and returns the ID of the remaining active
cluster, $B$.

\subsubsection{Modifications to the Chain-Based Algorithm}
Lastly, we discuss how to modify the chain-based algorithm to obtain an
$O(m\log n)$ time HAC for \trianglelink{}. The algorithm is identical
to Algorithm~\ref{alg:nnchain} with the only differences being the representation
of the \nghheap{} data structures, and the merge routine. Specifically,
the call to Merge on Line~\ref{line:nnchainmerge} uses the \textsc{FastMerge}
algorithm (Algorithm~\ref{alg:randomizedmerge}). As Lemma~\ref{lem:randchaingraphstate}
shows, after a merge, the state of the \nghheap{} data structures corresponds
to the state of the current graph induced by the active clusters, and thus
we do not have to modify \bestedgenoarg{}.

\begin{lemma}\label{lem:randchaingraphstate}
After a call to Algorithm~\ref{alg:randomizedmerge}, the adjacency information
stored in the \nghheap{}s (both \meldheap{A} and \meldtab{A}) of all active clusters $A$ is correct.
\end{lemma}
\begin{proof}
The proof is by induction. Consider the $k$-th merge between two vertices
$A$ and $B$, and assume that the claim holds before this merge. Assume without
loss of generality that $A$ is deactivated and $B$ remains active.
The only clusters affected by the merge are $\{B\} \cup \{C \in N(A)\}$, since all
neighbors in $N(B) \setminus N(A)$ have an edge to $B$ with the same weight as
before the merge.

First, we show that Algorithm~\ref{alg:randomizedmerge} correctly updates
the edges incident to $B$. The only edges that experience weight change are
those in $N(A) \cap N(B)$, which the algorithm detects when performing \textsc{T-Merge}
on Line~\ref{randmerge:mergetables}. For each neighbor $C$ in $N(A) \cap N(B)$,
it deletes $C$ from both \meldheap{A} and \meldheap{B} (Line~\ref{randmerge:deletelocs})
and reinserts $C$ into \meldheap{B} with the correct weight. Finally, the location
corresponding to $C$ is updated in \meldtab{B}. Note that \textsc{T-Merge} also
sets the weight of $C$ correctly in \meldtab{B}. The remaining affected edges are
new neighbors of $B$, which are correctly labeled and stored in \meldtab{B} and
\meldheap{B}.

Second, we show that Algorithm~\ref{alg:randomizedmerge} correctly updates the
\nghheap{}s for $C \in N(A)$. It processes these neighbors in the for-loop on
Line~\ref{randmerge:fornghs}. If the neighbor $C$ is not in $N(A) \cap N(B)$, it
just updates the cluster-id from $A$ to $B$ in \meldheap{C}, and leaves the location
in \meldtab{C} unchanged. Otherwise, for $C \in N(A) \cap N(B)$, it deletes $A$ and
$B$ from \meldheap{C}, updates the weight of the edge using $\mathcal{L}$ and reinserts
$B$ into \meldheap{C} (similarly for \meldtab{C}). Therefore, all of the neighbors
$C \in N(A)$ reference $B$ after Algorithm~\ref{alg:randomizedmerge} finishes.
\end{proof}

Using Lemma~\ref{lem:randchaingraphstate}
we have that the state of the \nghheap{} data structures correspond to
the current state of the graph induced by the active clusters after
performing a merge operation. Combining the fact that the \nghheap{} data is always
correct after a merge with the existing proof for the correctness of the
chain-based algorithm suffices to show that our randomized implementation
is correct. Next, we show that our approach is also efficient.

\begin{restatable}{theorem}{frameworkchainbased}
There is a randomized implementation of the chain-based algorithm
that runs in $O(m\log n)$ time in expectation for any \trianglelink{} $\mathcal{L}$.
\end{restatable}
\begin{proof}
We follow the proof of Theorem~\ref{thm:runtime} and separately account
for the cost of the merge steps, and the cost of the remaining steps in the
algorithm.

The algorithm performs $n-1$ merge operations, with a total merge-cost
of $O(m\log n)$ using Lemma~\ref{lem:mergecost}. To translate this cost
measure to the time-complexity measure, we examine the two types of operations
done inside of the \textsc{FastMerge} algorithm (Algorithm~\ref{alg:randomizedmerge}).
%\begin{enumerate}

The first type of updates are those done on the hash-tables, \meldtab{C} for a
cluster $C$. Over all merges, there are $O(m\log n)$ such operations, which each
cost $O(1)$ time in expectation, and thus the overall cost of the hash-table updates
are $O(m\log n)$ in expectation.

The second type of updates are done on the heaps. First, the cost of
melding the two heaps is $O(1)$ amortized using lazy Fibonacci heaps
and $O(\log n)$ using Leftist or eager Fibonacci heaps. In either
case, the overall cost of the meld operations is at most $O(n\log
n)$. The remaining heap operations can be broken up further into two
categories of heap updates which update the cluster-ids and weights of
edges in the heaps.
\begin{enumerate}
\item Updates that only affect the cluster-ids of an edge cost $O(1)$
time each in expectation, since they are done by looking up the
location of the edge in the hash-table, and updating the id of this
element in $O(1)$ time.\label{heap:type1}

\item Updates that change the weights of edges in the heap
are more costly since they require deleting and reinserting elements
from the heap, and deleting an element costs $O(\log n)$ time.
However, updating the weight of an edge is done only when two clusters
$A,B$ merge and both $A,B$ have an edge to a neighbor cluster $C$.
We observe that we can charge the cost of this step to one of the
original edges in the graph, and that each original edge is charged at
most once. \label{heap:type2}
\end{enumerate}
For heap updates of type (\ref{heap:type1}), the cost is thus $O(m\log
n)$ time in expectation. For heap updates of type (\ref{heap:type2}),
there are at most $O(m)$ of these updates, and each costs $O(\log n)$
time for a total cost of $O(m\log n)$ time.

Finally, the remaining steps in the algorithm outside of the merge
steps are $O(n)$ \bestedgenoarg{} queries, which each cost $O(\log n)$
time for a total cost of $O(n\log n)$ time. Thus, the overall
time-complexity of the algorithm is $O(m\log n)$ in expectation, as
desired.

\end{proof}

\subsection{Exact Unweighted Average-Linkage}

First, we present again the two key subroutines that make up our exact
unweighted average-linkage algorithm (Algorithms~\ref{alg:flipedge}
and Algorithm~\ref{alg:updateorientation}).

\begin{algorithm}\caption{$\textsc{FlipEdge}(A, B)$}\label{alg:flipedge}
    \begin{algorithmic}[1]
        \Require{Edge oriented from active clusters $A$ to $B$.}
        \Ensure{Edge oriented from $B$ to $A$.}
        \State $w = \hacweight{A}{B}$. [true weight of the edge]
        \State Update the edge $(A, B, w)$ in \heap{A}.
    \end{algorithmic}
\end{algorithm}

\begin{algorithm}\caption{$\textsc{UpdateOrientation}(A, B, \mathcal{EO})$}\label{alg:updateorientation}
    \begin{algorithmic}[1]
        \Require{Active clusters $A$ and $B$, dynamic orientation structure $\mathcal{EO}$.}
        \State $(A,B) = (\argmin(\degree{A}, \degree{B}), \argmax(\degree{A}, \degree{B})$.
        \State For each $C \in N(A)$, delete $(C, A)$ and insert $(C, B)$ into the orientation data structure. Edge flips are handled using \textsc{FlipEdge} [Algorithm~\ref{alg:flipedge}]
        \For {each edge $(B, C)$ oriented out of $B$}
            \State $w = \hacweight{B}{C}$. [true weight of the edge]
            \State Update the edge $(B, C, w)$ in \heap{C}.
        \EndFor
    \end{algorithmic}
\end{algorithm}

Let $\hat{G}$ denote the directed graph induced by the active
clusters, with edges oriented according to the edge-orientation
$\mathcal{EO}$. We start by proving a lemma that helps prove that our
exact unweighted average-linkage algorithm is correct.

\begin{lemma}\label{exactavg:maintainedgraph}
After a call to Algorithm~\ref{alg:merge}, for each active cluster
$A$, the weights of all edges directed towards $A$ in $\hat{G}$ are
set correctly in \heap{A}.
\end{lemma}
\begin{proof}
The proof is by induction. Consider the $k$-th merge between two
clusters $A$ and $B$, and assume that the claim holds before this
merge. Assume without loss of generality that $A$ is deactivated by
this merge, and $B$ remains active. Recall that the merge algorithm
will also invoke Algorithm~\ref{alg:updateorientation}, which updates
the orientation by remapping edges that are incident to the
deactivated cluster from the maintained orientation $\mathcal{EO}$,
and that Algorithm~\ref{alg:flipedge} is invoked each time the dynamic
edge-orientation algorithm flips an edge.

For any edges that are flipped during the execution of
Algorithm~\ref{alg:updateorientation}, the weights of these edges are
correctly set in the heap of the cluster that the edge now points to
(the \emph{head} of this edge). The reason is that
Algorithm~\ref{alg:flipedge} is invoked upon each edge flip and this
algorithm computes the correct weight of the edge and updates the
weight in the heap of the head of this edge. This accounts for all
edges that are flipped by the orientation algorithm when deleting all
$(A,C)$ (undirected) edges and reinserting them as $(B,C)$ edges.

The only remaining edges which may not have updated their
out-neighbors are new edges incident to $B$ that are oriented out of
$B$. Thus, Algorithm~\ref{alg:updateorientation} maps over all edges
oriented out of $B$ and sets the weight of these edges correctly in the
heap of the head of this edge.

We have accounted for (i) all edges whose orientation flips due
to the deletions and insertions in the merge and (ii) the new edges
oriented out of $B$. Finally, by assumption, the remaining edges have
their correct weights set in the heaps of the head of these edges, and
thus all active clusters $A$ have the correct weights for edges that
point to them.
\end{proof}

\begin{restatable}{theorem}{exactavglinkcomplexity}
The exact average-linkage algorithm is correct and runs in $\tilde{O}(n\sqrt{m})$
time for arbitrary graphs.
\end{restatable}
\begin{proof}
To show correctness, by Lemma~\ref{exactavg:maintainedgraph}, we have
that after the $i$-th merge, the state of each active cluster's heap
is correct for all but the $O(\alpha_i)$ edges that are oriented out
of this cluster. Since before performing a \bestedgenoarg{}
computation, the algorithm maps over all $O(\alpha_i)$ of these edges
and updates the weight of these edges in its heap to the correct
weight, all of the edges in its heap have the correct weight, and thus
the cluster selects the best-edge incident to it. The correctness of
the overall algorithm can now by obtained by combining this proof with
the existing proof for the correctness of the nearest-neighbor chain
algorithm.

Next, we analyze the time-complexity of our algorithm.
We first bound the extra cost incurred by maintaining the \dynorient{}
data structure $\mathcal{EO}$ over the course of the algorithm.
Using the \dynorient{} data structure of Henzinger et
al.~\cite{henzinger20dynamic}, we obtain an amortized cost of
$O(\log^2 n)$ for each edge insertion and deletion.
The \dynorient{} data structure is only updated and used during the
\merge{} and \bestedge{} operations. Consider the $i$-th such
operation, and let $G_i$ be the graph induced by the current
clustering at the time of this operation.

Using Lemma~\ref{lem:mergecost}, the total number of merge operations
is at most $O(m\log n)$. For each of these operations, we have to
perform an edge insertion and deletion, which could translate to a
total $O(m\log^3 n)$ edge flips. Each edge flip also requires $O(\log
n)$ time to update the weight of the edge in the heap of the new head
of this edge. Thus the total cost of the edge insertions, deletions,
and flips is $O(m\log^4 n)$ over the course of the entire algorithm.

The merge algorithm also processes the edges incident to the remaining
active cluster, $B$. The cost for this step is $O(\alpha_i \log n)$,
since there are $O(\alpha_i)$ edges oriented out of $B$ and we pay
$O(\log n)$ to perform a heap-update for each one. Since there are
$n-1$ merges, the overall cost for this step is $O(\log
n\sum_{i=1}^{n-1} \alpha_i)$ over the course of the entire algorithm.

Lastly, the cost for performing the \bestedgenoarg{} operation is
$O(\alpha^{*}_i)$ where $\alpha^{*}_i$ is the current arboricity at
the time of the $i$-th \bestedgenoarg{} operation. Note that there may
be many \bestedgenoarg{} operations performed before a merge is
performed. Since there are $2n-2$ \bestedgenoarg{} operations, the
overall cost is $O(\sum_{i=1}^{2n-2} \alpha^{*}_i)$.

By bounding each $\alpha_i, 1 \leq i \leq n-1$ and $\alpha^{*}_j, 1
\leq j \leq 2n-2$ above as $\alpha_{\max} \leq \sqrt{m}$, the maximum
arboricity of the graph over the entire sequence of merges, the
overall time-complexity of the algorithm is
\[
  O(m\log^{4} n + n\log n \cdot \alpha_{\max}) = \tilde{O}(n \sqrt{m}).
\]

\end{proof}

\subsection{Approximate Unweighted Average-Linkage}

We start by recalling the notion of approximation and the invariant
used in our approximation algorithm.
An \emph{$\epsilon$-close HAC algorithm} is an algorithm which only
merges edges that have similarity at least
$(1-\epsilon) \cdot \mathcal{W}_{\max}$ where $\mathcal{W}_{\max}$
is the largest weight currently in the graph~\cite{48657}.

The idea of our algorithm is to maintain an extra counter for each
cluster which stores the size the cluster had at the last time
that the algorithm updated \emph{all} of the incident edges of
the cluster. Call this variable the \emph{staleness}, $\staleness{A}$, of a
given cluster $A$.
Recall that the size of a cluster $|A|$ is defined to be the number 
of initial  (singleton) clusters that it contains.
Our algorithm maintains the following invariant:
\begin{invariant}\label{inv:staleness}
For any active cluster $A$, $|A| < (1+\epsilon)\staleness{A}$.
\end{invariant}
Let the stored similarity of an edge $(u,v)$ in the neighborhood be denoted \storedwgh{u,v} and the true similarity of this edge be \truewgh{u,v}. The next lemma bounds the maximum error an algorithm maintaining Invariant~\ref{inv:staleness} can observe for an edge incident to an active cluster.
\begin{restatable}{lemma}{weightsapproxclose}\label{lem:weightsapproxclose}
Let $e=(U, V)$ be an edge in the neighborhood of an active cluster $U$. Then, $(1+\epsilon)^{-2}\storedwgh{U,V} \leq \truewgh{U,V}.$
\end{restatable}
\begin{proof}
There are two ways that the similarity of the $(U,V)$ edge can change as
the algorithm merges clusters:

\begin{enumerate}
\item By a merge which adds parallel edges to $\cut{U, V}$ (e.g.,
if $U$ merges with a cluster $Z$ which is also connected to $V$,
thereby increasing the total similarity of edges crossing the
cut). \label{lab:typeone}
\item By a merge to $|U|$ or $|V|$ that does not affect the total
similarity of edges going across $\cut{U, V}$ (e.g., if $U$ merges with
a cluster $Z$ that is not connected to $V$). \label{lab:typetwo}
\end{enumerate}

If a Type~\ref{lab:typeone} update occurs, then the similarity of this
edge will be set to the true value upon this update, since the
algorithm will update the similarities of every edge in the intersection of
the merge. Therefore, we can ignore these updates when trying to bound
the maximum drift between the true and stored similarities.
On the other hand, we could have many Type~\ref{lab:typetwo} updates
occur. In this case, the similarity of this edge will not be updated
unless Invariant~\ref{inv:staleness} becomes violated for either $U$
or $V$. Therefore, we have the following upper bounds on the maximum
size of $U$ and $V$, namely that $|U| < (1+\epsilon)S(U)$ and $|V| <
(1+\epsilon)S(V)$.

In the worst case, the stored similarity could have used $S(U)$ and $S(V)$ to
normalize (since the similarity must have been updated when these stored
similarities were set), and the true similarity could use $(1+\epsilon)S(U)$
and $(1+\epsilon)S(V)$. Since the sum term in the similarity equation
doesn't change (since there are no Type~\ref{lab:typeone} updates), we
have that
\[
\storedwgh{U,V} \leq  \frac{1}{S(U) \cdot S(V)} \cdot  \sum_{(u,v) \in \cut{U,V}} w(u,v)
\]
and therefore
\[
\truewgh{U,V} \geq \frac{1}{(1+\epsilon)^2\cdot S(U)\cdot S(V)} \cdot \sum_{(u,v) \in \cut{U,V}}
\]
Therefore, the true similarity is at most a $(1+\epsilon)^{-2}$ factor
smaller than the stored similarity.
\end{proof}

\begin{restatable}{theorem}{appxalgorithmbound}
There is an $\epsilon$-close HAC algorithm for the average-linkage
measure that runs in $O(m\log^2 n)$ time.
\end{restatable}
\begin{proof}
First we show that our algorithm is $\epsilon$-close.
Let $\delta$ be the closeness parameter used internally in the
algorithm, which we will set shortly. For each $(U,V)$ edge, since the
simialrities in the algorithm only decrease, we have that $\truewgh{U,V}
\leq \storedwgh{U,V}$.  Combining this fact with
Lemma~\ref{lem:weightsapproxclose}, we have
that
\[
  (1+\delta)^{-2}\storedwgh{U,V} \leq \truewgh{U,V} \leq \storedwgh{U,V}
\]
for all active edges $U,V$.

Next, consider a merge step in the algorithm which merges two clusters
$A,B$. We have that $\storedwgh{A,B}$ is the largest stored similarity
among any active cluster in the graph. We also have that
$(1+\delta)^{-2}\storedwgh{A,B} \leq W_{\max}$ where $W_{\max}$ is the
current maximum similarity in the graph, since otherwise the stored
similarity corresponding to $W_{\max}$ would be larger than
\storedwgh{A,B} and thus $(A,B)$ would not be the edge selected from
the global heap, $H$. Therefore our approach yields a
$(1 - (1+\delta)^{-2})$-close algorithm, and by setting 
$\delta = \sqrt{1/(1-\epsilon)} - 1$ we obtain an $\epsilon$-close algorithm.

Lastly, we show that the algorithm runs in $O(m\log^2 n)$ time for any
given constant $\epsilon$. Other than the extra work done to update
stale clusters, the algorithm is identical to the heap-based algorithm
from our framework. To bound the work done for stale clusters, observe
that each cluster can become stale at most $O(\log_{1+\delta}n ) = O(\log_{1+\epsilon} n)$
times, and performs $O(\log n)$ work per incident edge each time it
becomes stale. Since each of the original $m$ edges is associated with
at most two active clusters at any point in the algorithm, the overall
time-complexity of updating stale vertices is $O(m \log_{1+\epsilon} n
\log n) = O(m\log^2 n)$ for any constant $\epsilon$. Combining this
with the time-complexity of the heap-based algorithm completes the
proof.
\end{proof}

\section{Experimental Evaluation}\label{sec:exps}
\myparagraph{Graph Data}
We list information about graphs used in our experiments in
Table~\ref{table:sizes}. 
\emph{com-DBLP} {\bf (DB)} is a co-authorship network sourced from the
DBLP computer science bibliography\footnote{Source: \url{https://snap.stanford.edu/data/com-DBLP.html}.}. \emph{YouTube} {\bf (YT)}
is a social-network formed by user-defined groups on the YouTube
site\footnote{Source: \url{https://snap.stanford.edu/data/com-Youtube.html}.}. 
\emph{Skitter} {\bf(SK)} is an internet topology graph generated from traceroutes\footnote{Source: \url{https://snap.stanford.edu/data/as-Skitter.html}.}.
\emph{LiveJournal} {\bf(LJ)} is a directed graph of the
social network\footnote{Source: \url{https://snap.stanford.edu/data/soc-LiveJournal1.html}.}.
\emph{com-Orkut} {\bf (OK)} is an undirected
graph of the Orkut social network\footnote{Source: \url{https://snap.stanford.edu/data/com-Orkut.html}.}.
These graphs are sourced from the SNAP dataset~\cite{leskovec2014snap}.

Another family of graphs that we consider are generated from point
datasets by using an approximate $k$-NN graph construction. All of the point datasets that we use can be found in the sklearn.datasets package\footnote{For more detailed information see \url{https://scikit-learn.org/stable/datasets.html}.}.
We note that the large real-world graphs that we study are not weighted, and 
so we set the similarity of an edge $(u,v)$
to $\frac{1}{\log(\degree{u} + \degree{v})}$.
We symmetrized all directed graph inputs considered in this paper.

\begin{table}\footnotesize
\centering
\centering
\caption{\small Graph inputs, including vertices and
  edges.}
\smallskip{}
\begin{tabular}[!t]{lrr}
\toprule
{Graph Dataset} & Num. Vertices & Num. Edges \\
\midrule
{\emph{com-DBLP} {\bf (DB)} }           & 425,957          &2,099,732\\
{\emph{YouTube-Sym} {\bf (YT)} }        & 1,138,499        &5,980,886 \\
{\emph{Skitter-Sym} {\bf (SK)} }        & 1,696,415        &22,190,596 \\
{\emph{LiveJournal-Sym} {\bf(LJ)} }    & 4,847,571        &85,702,474 \\
{\emph{com-Orkut      } {\bf(OK)} }    & 3,072,627        &234,370,166
\end{tabular}
\label{table:sizes}
\end{table}

\newcommand{\STAB}[1]{\begin{tabular}{@{}c@{}}#1\end{tabular}}
\begin{table*}\footnotesize
\centering
\caption{\small Adjusted Rand-Index (ARI) and Normalized Mutual Information (NMI) scores
of our graph-based HAC implementations (columns 2--5) versus the HAC implementations from sklearn (columns 6--9).
The scores are calculated by evaluating the clustering generated by each cut of the generated 
dendrogram to the ground-truth labels for each dataset. Our graph-based implementations run
over an approximate $k$-NN graph with $k = 50$.}
\smallskip{}
\begin{tabular}{@{}cl ccccc|cccc}
\toprule
& {Dataset} & Single & Complete & WPGMA & Apx-Avg & Avg & Sk-Single & Sk-Complete & Sk-Avg & Sk-Ward \\
\midrule
\multirow{5}{*}{\STAB{\rotatebox[origin=c]{90}{{ARI}}}}
& {\emph{iris}}   & 0.702 & 0.462 & 0.605 & {\bf 0.759} & {\bf 0.759}    & 0.714 & 0.642 & {\bf 0.759} & 0.731 \\
& {\emph{wine}}   & 0.297 & 0.286 & 0.317 & 0.331 & 0.331    & 0.297 & {\bf 0.370} & 0.351 & 0.368 \\
& {\emph{digits}} & 0.661 & 0.133 & 0.500 &  0.876 & {\bf 0.880}    & 0.661 & 0.478 & 0.689 & 0.812 \\
& {\emph{cancer}} & {\bf 0.561} & 0.543 & 0.539 & 0.489 & 0.489    & {\bf 0.561 }& 0.464 & 0.537 & 0.406 \\
& {\emph{faces}}  & 0.467 & 0.438 & 0.480 & 0.508 & 0.508    & 0.467 & 0.471 & 0.529 & {\bf 0.608} \\
\midrule
\multirow{5}{*}{\STAB{\rotatebox[origin=c]{90}{{NMI}}}}
& {\emph{iris}}   & 0.733 & 0.641 & 0.733 & {\bf 0.805} & {\bf 0.805}   & 0.761 & 0.722 & {\bf 0.805} & 0.770 \\
& {\emph{wine}}   & 0.410 & 0.388 & 0.387 & 0.427 & 0.427   & 0.417 & {\bf 0.463} & 0.448 & 0.448 \\
& {\emph{digits}} & 0.772 & 0.572 & 0.713 & 0.900 & {\bf 0.902}    & 0.778 & 0.729 & 0.838 & 0.868 \\
& {\emph{cancer}} & 0.316 & 0.359 & 0.384 & {\bf 0.460} & {\bf 0.460}    & 0.385 & 0.442 & 0.456 & 0.446 \\
& {\emph{faces}}  & 0.847 & 0.846 & 0.857 & 0.859 & 0.859    & 0.856 & 0.855 & 0.867 & {\bf 0.871}
\end{tabular}
\label{table:sklearnaccuracy}
\end{table*}

%\appendix
%\input{inputs/proofs}
%\input{inputs/flight}

\end{document}